\documentclass[11pt,letterpaper]{article}

\usepackage[margin=1in]{geometry}

\usepackage{graphicx}
\usepackage{graphics}

\usepackage{amsmath}
\usepackage{amsfonts}
\usepackage{amssymb}
\usepackage{enumerate}
\usepackage{xcolor}
\usepackage{tikz}
\usepackage{textcomp} 
\definecolor{Green}{rgb}{0.0, 0.5, 0.0}
\usepackage{appendix}
\usepackage{mathtools}
\mathtoolsset{showonlyrefs=true}

\newtheorem{theorem}{Theorem}[section]

\newtheorem{corollary}[theorem]{Corollary}

\newtheorem{definition}[theorem]{Definition}
\newtheorem{example}[theorem]{Example}

\newtheorem{lemma}[theorem]{Lemma}

\newtheorem{proposition}[theorem]{Proposition}

\newtheorem{remark}[theorem]{Remark}

\newenvironment{proof}[1][Proof]{\textbf{#1.} }{\ \rule{0.5em}{0.5em}}

\newcommand{\Var}{\operatorname{Var}}

\newcommand{\Prob}{\mathbb{P}}

\newcommand{\supp}{\operatorname{supp}}
\newcommand{\Inf}{\operatorname{Inf}}
\newcommand{\Ibin}{\mathtt I_{\mathrm{bin}}}
\newcommand{\Ilead}{\mathtt I_{\mathrm{leader}}}

\newcommand{\numbercellong}[2]

\usepackage{dirtytalk}

\usepackage[english]{babel}

\usepackage{amsmath}
\usepackage{graphicx}
\usepackage[colorlinks=true, allcolors=blue]{hyperref}

\title{Entropy Bounds for Local Coordination and Graph Amenability}
\author{Ron Peretz\thanks{Economics Department, Bar-Ilan University, Israel. ron.peretz@biu.ac.il} \and Dean Kraizberg\thanks{Department of Mathematics, Bar-Ilan University. kraizbd@biu.ac.il}}
\date{\today}

\begin{document}
\maketitle
\renewcommand{\thefootnote}{\fnsymbol{footnote}}

\renewcommand{\thefootnote}{\arabic{footnote}}

\begin{abstract}
We study local coordination on finite graphs, equivalently, the ability of randomized local protocols with independent private randomness to produce highly correlated outputs at neighboring vertices. Hutchcroft, Rospuskova, and Tamuz~\cite{TOOsqrtbound2026} showed that if such local coordination can be achieved with low cost, then the graph admits a low-cost leader profile, or equivalently an appropriate bounded-radius decomposition, with a square-root loss in the cost.

We substantially improve this bound for unbiased binary outputs. If the average disagreement is $\varepsilon$, we construct a leader profile of cost $O(\varepsilon\log(1/\varepsilon))$. The key ingredient is an information-theoretic influence measure: after fixing an ordering of the independent information sources, we assign to each binary output the increments of mutual information revealed along that ordering. For two outputs that disagree with probability $p$, the corresponding measures are within binary entropy $h(p)$ in total variation. Combined with a grand coupling theorem, this yields the improved bound.

We also show that a universal linear bound is impossible. We construct graphs with binary local profiles of cost $O(\varepsilon)$ for which every leader profile has cost
\[
\Omega\!\left(\varepsilon\sqrt{\log(1/\varepsilon)}\right).
\]
The construction uses weighted majority functions and can be implemented as an action-symmetric Nash equilibrium, giving a negative answer to the question of~\cite{TOOsqrtbound2026} whether leader equilibria are always optimal up to a constant factor.

Finally, we show that the binary assumption is essential: for general real-valued outputs, the square-root dependence of the HRT theorem is sharp up to a universal constant.
\end{abstract}


\newpage
\section{Introduction}

Consider a network in which every vertex has an independent source of private randomness. 
Information can be communicated locally for \(r\) rounds, after which each vertex produces a binary output. Regardless of the communication protocol, the output at a vertex can depend only on randomness originating within its radius \(r\) neighborhood. A basic question is how well such a local randomized process can coordinate its outputs across the edges of the network.

We measure performance by disagreement between neighboring vertices. We require each binary output to be unbiased, since otherwise perfect agreement can be achieved trivially by having every vertex choose the same predetermined value. Thus the problem is to understand how strongly local computation can correlate neighboring unbiased bits when the underlying random sources are independent.

The same problem arises naturally in models of coordination on social networks. Agents wish to match the actions of their neighbors and may communicate only locally before acting. Examples include the adoption of conventions, standards, or shared terminology. In this interpretation, disagreement along an edge represents a coordination failure. From the distributed computing perspective, however, the underlying issue is more general: it concerns the correlations that can be generated by randomized local protocols subject only to locality of information.

Hutchcroft, Rospuskova, and Tamuz~\cite{TOOsqrtbound2026} recently showed that this local coordination problem is closely connected to the geometry of the network. Their notion of amenability, also known in the mathematics literature as hyperfiniteness, asks whether the graph can be partitioned into communities of bounded radius after removing only a small fraction of its edges. Such a partition immediately yields a simple local coordination scheme: every community follows a single nearby random source, so disagreement can occur only across the removed edges. We refer to such schemes as leader profiles.

The main result of HRT is a converse. Roughly, if a local coordination profile has average inefficiency \(\varepsilon\), then there exists a leader profile with inefficiency of order \(\sqrt{\varepsilon}\). Thus, up to a square root loss, the full power of local randomized computation can be replaced by a particularly simple mechanism in which every vertex merely follows one nearby source. Their proof associates with each output a probability distribution over its possible information sources, using Shapley values, and then applies a grand coupling theorem to convert these distributions into common leaders.

This raises a natural quantitative question. How much more powerful can an arbitrary local randomized protocol be than a leader protocol? In particular, HRT ask whether the square root loss can be replaced by a constant factor in the binary setting.

We show that the answer lies strictly between these two possibilities: the square root loss can be substantially improved, but that a universal constant factor is impossible. Let \(\mathtt I_{\mathrm{bin}}(G,r)\) denote the minimum inefficiency attainable by an unbiased binary radius \(r\) profile, and let \(\mathtt I_{\mathrm{leader}}(G,r)\) denote the corresponding minimum over leader profiles. Writing
\[
h(p)=-p\log_2p-(1-p)\log_2(1-p)
\]
for the binary entropy function, we prove
\[
\mathtt I_{\mathrm{leader}}(G,r)
\leq
2h\!\left(\frac{1}{2}\mathtt I_{\mathrm{bin}}(G,r)\right)
=
O\!\left(
\mathtt I_{\mathrm{bin}}(G,r)
\log\frac{1}{\mathtt I_{\mathrm{bin}}(G,r)}
\right).
\]
Thus arbitrary local binary coordination can be replaced by leader coordination with only a logarithmic loss.

The proof replaces the Shapley influence distribution~\cite{Owen2014,OwenPrieur2017} which HRT use in their square-root bound proof by a simpler information-theoretic construction. Fix an ordering of the independent information sources. For every binary output \(X\), we record how much mutual information about \(X\) is revealed as the sources are exposed one by one. The resulting increments form a probability distribution over the information sources. We show that for two unbiased binary outputs \(X\) and \(Y\), the total variation distance between these distributions is bounded by the binary entropy of \(\Pr(X\neq Y)\). Combining this estimate with the same grand coupling argument~\cite{KleinbergTardos2002,AngelSpinka2019} used by HRT yields the entropy bound above.

We then show that the remaining logarithmic loss cannot in general be eliminated completely. For every sufficiently small \(\varepsilon\), we construct graphs admitting a local unbiased binary profile of cost \(O(\varepsilon)\) for which every leader profile has cost at least
\[
\Omega\!\left(
\varepsilon\sqrt{\log(1/\varepsilon)}
\right).
\]
The construction uses weighted majority functions and the geometry of the hyperplane arrangement that parametrizes them. Moreover, the inexpensive binary profile can be implemented as an action symmetric Nash equilibrium. This gives a negative answer to the question raised by HRT of whether leader equilibria are always optimal up to a universal constant factor.

Finally, we show that the binary restriction is essential for any improvement of the original theorem. For general unrestricted real valued outputs, the square root dependence obtained by HRT is sharp up to a universal constant.

\paragraph{Information constraints beyond graph distance.}
It is useful to separate the graph on which disagreement is measured from the information available to each vertex. We therefore formulate our main upper bound in a more general model. Let \(N\) be the set of players and \(K\) a set of independent information sources. Each player \(i\) has access to an arbitrary nonempty subset \(B(i)\subseteq K\). We call the resulting access structure an \emph{information skeleton}. The original radius \(r\) model is recovered by taking \(K=N\) and \(B(i)=B_r(i)\).

This abstraction makes clear that the upper bound does not depend on any special property of graph distance. If \(\mathtt I_{\mathrm{bin}}(G,\mathfrak S)\) and \(\mathtt I_{\mathrm{leader}}(G,\mathfrak S)\) denote the optimal binary and leader inefficiencies for an information skeleton \(\mathfrak S\), respectively, then our main theorem states
\[
\mathtt I_{\mathrm{leader}}(G,\mathfrak S)
\leq
2h\!\left(\frac12 \mathtt I_{\mathrm{bin}}(G,\mathfrak S)\right).
\]
The graph \(G\) determines which pairs of outputs we wish to coordinate, while the skeleton \(\mathfrak S\) determines which independent sources may be used to produce each output.

\paragraph{The entropy argument.}
The key estimate underlying the upper bound is local and independent of the graph. Let \(Z_1,\ldots,Z_k\) be independent sources and let \(X\) be an unbiased binary variable measurable with respect to them. Fix an arbitrary ordering of the sources and define
\[
f_X(j)=I(X;Z_1,\ldots,Z_j).
\]
Since \(f_X(0)=0\) and \(f_X(k)=H(X)=1\), its increments define a probability measure \(df_X\) on the sources. Moreover, if \(X\) depends only on a subset \(B\subseteq K\), then \(df_X\) is supported on \(B\).

For two unbiased binary variables \(X\) and \(Y\), with
\[
p=\Pr(X\neq Y),
\]
we prove
\[
\|df_X-df_Y\|_{\mathrm{TV}}\leq h(p).
\]
The proof is a direct application of the chain rule for mutual information. We compare both measures with the information increments of the pair \((X,Y)\); the resulting distances are controlled by \(H(X\mid Y)\) and \(H(Y\mid X)\), both equal to \(h(p)\). The estimate is sharp.

Following HRT, we then apply the grand coupling theorem of Kleinberg and Tardos~\cite{KleinbergTardos2002}, in the formulation of Angel and Spinka~\cite{AngelSpinka2019}, to the family of measures associated with all vertices. The coupled samples serve as leaders. Neighboring vertices receive different leaders with probability controlled by the total variation distance of their information measures. Averaging over edges and using concavity of binary entropy gives the theorem.

One can view this step as a rounding argument. The information measures provide fractional choices of admissible leaders, and the grand coupling simultaneously rounds them to actual leaders while approximately preserving agreement across every edge.

\paragraph{Why a constant factor is impossible.}
The lower bound is qualitatively different. Its basic objects are odd monotone weighted majority functions on \(k\) unbiased bits. Weight vectors defining the same Boolean function form regions of the positive orthant separated by hyperplanes. We construct an auxiliary graph whose vertices are these majority functions and whose edges join functions corresponding to adjacent regions. Neighboring functions disagree on only one antipodal pair of inputs, and hence with probability \(2^{-(k-1)}\).

The graph used in the lower bound contains many copies of each majority function, together with special vertices representing the input coordinates. A copy of a function observes the coordinates on which that function depends and outputs the corresponding weighted majority. Edges between copies of adjacent functions are sampled with probabilities derived from the geometry of the hyperplane arrangement. The resulting binary profile therefore has cost \(O(2^{-k})\).

The main difficulty is to show that every leader assignment nevertheless cuts many edges. The geometric construction of the edge probabilities has the following consequence. For any collection of probability measures \(\mu_f\), each supported on the coordinates relevant to the majority function \(f\),
\[
\sum_{\{f,g\}}p_{fg}\|\mu_f-\mu_g\|_{\mathrm{TV}}
\geq \frac12.
\]
This follows by increasing one coordinate of a random weight vector from zero to its final value and recording the facets crossed along the path. At the beginning of the path that coordinate is irrelevant, while at the endpoint the masses assigned to all coordinates sum to one. Summing these changes over the coordinates yields the inequality.

The total weight of the edges in this auxiliary graph is of order \(2^k/\sqrt{k}\). After replication and concentration, the preceding inequality forces every leader assignment to cut a fraction of order \(\sqrt{k}\,2^{-k}\) of the edges, while the binary profile has cost \(O(2^{-k})\). Setting \(\varepsilon\asymp2^{-k}\) gives the separation
\[
\mathtt I_{\mathrm{leader}}
=
\Omega\!\left(
\varepsilon\sqrt{\log(1/\varepsilon)}
\right).
\]
The concentration argument is uniform over all leader assignments, including assignments chosen after the random graph has been realized.

The same construction can be implemented within the strategic model of HRT. The weighted majority profile can be chosen to be an action symmetric Nash equilibrium. Hence the separation is not merely between arbitrary binary profiles and leader profiles: it also gives a negative answer to HRT's question of whether leader equilibria approximate all equilibria within a universal constant factor.

\paragraph{Further results.}
We obtain several additional results that help clarify the remaining gap between our upper and lower bounds. First, returning to the Shapley influence distributions used by HRT, we show that for binary outputs their total variation distance also satisfies an \(O(\varepsilon\log(1/\varepsilon))\) bound. Thus the logarithmic improvement is not specific to our information increment construction, although our proof of the main theorem is considerably simpler. The proof for Shapley values uses a martingale representation together with the Burkholder inequality.

More importantly, we consider the entire core of the variance game associated with each output, rather than its Shapley value, which belongs to the core since the game is supermodular. Every measure in this core is supported on information sources available to the corresponding player and can therefore serve as a fractional leader distribution in the grand coupling argument. We prove that for two unbiased binary outputs \(X\) and \(Y\), the Hausdorff distance between their cores, in total variation, is at most \(4\Pr(X\neq Y)\). Hence the feasible sets themselves are linearly close. The obstacle to obtaining a universal linear upper bound is therefore a global selection problem: one must choose a single measure from every player's core so that neighboring choices are simultaneously close. On a tree there is no consistency obstruction, since one may root the tree and choose measures recursively along its edges, and the core estimate consequently yields a linear leader bound. For graphs with cycles, choices made along different paths may conflict. This perspective suggests that spanning tree methods, tree approximations, or related decompositions may lead to stronger bounds for special classes of graphs, although we do not pursue this direction here.

Finally, for every fixed communication radius \(r\) and degree bound \(\Delta\), we prove directly that general inefficiency and leader inefficiency are linearly related, with a constant depending only on \(r\) and \(\Delta\).

\section{The model}\label{sec:theModel}
Hutchcroft, Rospuskova, and Tamuz consider an optimization problem on a graph $G=(N,E)$, where the players $N$ try to coordinate ex-ante random binary actions through communication limited to distance $r>0$. The graph and the parameter $r$ define both the communication constraint and the objective function, minimizing expected number of disagreements between neighbors. It is both natural and helpful to separate the constraints and the objective function. In particular, the graph structure itself does not play a role in the constraints imposed.

\medskip

Let $N=\{1,\dots,n\}$ be a set of players, and $K=\{1,\dots,k\}$ be a set of sources of information. Each player $i\in N$ has access to a subset of the information sources $\emptyset \neq B(i)\subseteq K$. We call such a triple $\mathfrak  S=(N,K,B\colon N\to2^{K})$ an \emph{information skeleton}. A \emph{coordination profile} adapted to $\mathfrak S$ is a family of random variables $\pi=(X_i,Z_j)_{i\in N,j\in K}$, where the $(Z_j)_{j\in K}$ are mutually independent, each $X_i$ is measurable with respect to $Z_{B(i)}:=(Z_j)_{j\in B(i)}$, $\mathbb E[X_i]=0$, and  $\mathbb E[X_i^2]=1$. We say that the profile is binary if each $X_i$ assumes values in $\{+1,-1\}$. A binary profile is called a \emph{leader profile} if $\operatorname{Cov}(X_i,X_j)\in\{0,1\}$ for every $i,j\in N$. Given a (simple undirected) graph structure on the players, a coordination profile $\pi$ incurs a discoordination cost $\mathtt C(G,\pi)$ defined by
\[
\mathtt C(G,\pi):=\frac 1{|E|}\sum_{ij\in E}\left(1-\operatorname{Cov}(X_i,X_j)\right)\ .
\]
The inefficiency of $G$ with respect to $\mathfrak S$ is defined as 
\[
\mathtt I(G,\mathfrak S):=\inf_{\substack{\pi \text{ adapted }\\\text{to }\mathfrak S}} \mathtt C(G
,\pi)\ .
\]
Similarly, the \emph{binary inefficiency} and \emph{leader inefficiency}, $\mathtt I_{\mathrm{bin}}(G,\mathfrak S)$ and $\mathtt I_{\mathrm{leader}}(G,\mathfrak S)$, are respectively defined with the infimum taken over binary and leader coordination profiles. Taking the $(Z_j)_{j\in K}$ to be independent fair signs and
letting each player follow an arbitrary source in $B(i)$ gives
a leader profile of cost at most $1$. Consequently,
\[
\mathtt I(G,\mathfrak S)\leq \mathtt I_{\mathrm{bin}}(G,\mathfrak S)\leq \mathtt I_{\mathrm{leader}}(G,\mathfrak S) \le 1 .
\]
Leader inefficiency has the following equivalent combinatorial interpretation: it is the minimal proportion of edges that must be removed so that every connected component of the remaining graph has a common available information
source. Following \cite{TOOsqrtbound2026} we introduce the following definition.
\begin{definition}[\((\varepsilon,\mathfrak S)\)-amenability]
Let \(G=(N,E)\) be a finite graph and $\mathfrak S=(N,K,B)$ an information skeleton. We say that \(G\) is
\((\varepsilon,\mathfrak S)\)-amenable if there exists a partition of $N$ such that each part $C$ has $\bigcap_{i\in C}B(i)\neq\emptyset$, and such that the number of edges crossing between distinct parts is at most \(\varepsilon |E|\).
\end{definition}

Given a leader profile, partition $N$ into classes $C$ of players
whose outputs agree almost surely. Independence of the sources
implies that their common output is measurable, up to null sets,
with respect to $Z_{\bigcap_{i\in C}B(i)}$. Since this output has
variance one, $\bigcap_{i\in C}B(i)\neq\emptyset$. Outputs in
different classes have covariance zero, so the cost is exactly
the proportion of edges between classes.

Conversely, give each source an independent fair sign and let
each part of an admissible partition follow a chosen common
source. This produces a leader profile with cost at most the
proportion of edges between parts, since parts choosing the same
source merge. These two directions establish the stated combinatorial
characterization of leader inefficiency.
Since there are finitely many partitions, the infimum defining
$\mathtt I_{\mathrm{leader}}(G,\mathfrak S)$ is attained.

\medskip

The natural information skeleton used in~\cite{TOOsqrtbound2026} is associated with range $r>0$ communication in a social network $G=(N,E)$. It is defined by
\[
\mathfrak S(G,r):=(N,N,B_r)\, ,
\]
where $B_r(i):=\{j\in N:\mathrm{dist}_G(i,j)\leq r\}$ is the ball of radius $r$ around $i$. With tolerable abuse of notation we write 
\[
\mathtt I(G,r):=\mathtt I(G,\mathfrak S(G,r))
\]
and similarly for $\mathtt I_{\mathrm{bin}}$ and $\mathtt I_{\mathrm{leader}}$. With a similar notational abuse, a graph $G$ is called $(\varepsilon,r)$-amenable if it is $(\varepsilon,\mathfrak S(G,r))$-amenable, or equivalently, $\mathtt I_{\mathrm{leader}}(G,r)\leq \varepsilon$.

\begin{example}
    Let $C_n$ denote the cycle graph on the vertex set $\{1,\ldots,n\}$, where $\{i,j\}\in E$ if and only if $j = i+1 \quad \text{ or } \quad i-1 \mod(n)$. If $2r+1$ divides $n$, one may verify that $C_n$ is $(\frac 1 {2r+1} , r)$-amenable.
\end{example}

\section{Bounds of Amenability Parameter}

One of the main results in~\cite{TOOsqrtbound2026} is that low inefficiency in local coordination forces the underlying graph to be amenable. Theorem~2 in \cite{TOOsqrtbound2026} states that 
\[
\mathtt I_{\mathrm{leader}}(G,r)\leq \sqrt{8 \mathtt I_{\mathrm{bin}}(G,r)}\, ,
\]
for every finite graph $G$ and radius $r>0$. The proof of~\cite[Theorem~3]{TOOsqrtbound2026} in fact gives the following more general statement.
\begin{equation}\label{eq:HRT-general}
\mathtt I_{\mathrm{leader}}(G,\mathfrak S)\leq \sqrt{8 \mathtt I(G,\mathfrak S)}\, .    
\end{equation}

HRT suggest\footnote{They conjecture that leader equilibria are always optimal~\cite[Sections 6--7]{TOOsqrtbound2026}, which motivates them to ask whether the bound above can be improved.} that the bound in~\cite[Theorem 2]{TOOsqrtbound2026} can be improved, perhaps even to $C\mathtt I_{\mathrm{bin}}(G,r)$ for a universal constant $C$. Their conjecture is nearly true: we show that a linear bound is achievable up to a logarithmic factor (see Theorem~\ref{log factor bound amenability}). However, we also show that the conjectured universal linear bound cannot hold, by proving the existence of graphs for which leader strategies incur a square-root logarithmic loss relative to a specific binary profile that can be also implemented as equilibria (see Theorem~\ref{thm:counter to linearity}).

\medskip

We begin our discussion of the results by showing that the bound \eqref{eq:HRT-general} is tight up to a constant factor in the general non-binary setting (see Proposition~\ref{sqrt is sharp in general} below), and hence cannot be improved asymptotically.

\begin{proposition}\label{sqrt is sharp in general}
There exists a universal constant \(c>0\) such that for every \(0<\varepsilon<1\) there exist a graph \(G\) and an integer \(r>0\) such that $\mathtt I(G,r)<\varepsilon$, yet 
\[
\mathtt I_{\mathrm{leader}}(G,r) \ge c \sqrt{\varepsilon}\, .
\]
\end{proposition}

\medskip
\noindent\textbf{Proof overview.}
Consider the cycle graph $C_n$, where $n$ is a sufficiently large
multiple of $2r+1$, and denote
$d(i,j)=\mathrm{dist}_{C_n}(i,j)$.
Let the $Z_i$ be independent unbiased sign random variables, and define
\[
X_i=\sum_{j\in B_r(i)}c_{d(i,j)}Z_j,
\]
with $c_d$ proportional to $r+1-d$ and normalized so that
$c_0^2+2\sum_{d=1}^r c_d^2=1$.
The cost of this profile is at most $3/r^2$, while the optimal
amenability parameter is $1/(2r+1)$.
Taking $r=\lceil\sqrt{3/\varepsilon}\rceil+1$ gives the result.
The full proof is given in Section~\ref{sec:additional proofs}.$\ \square$

\medskip

Proposition~\ref{sqrt is sharp in general} shows that improving the universal bound requires restricting the class of profiles, for example to binary profiles.\footnote{In fact, the same entropy argument used in the proof of Theorem~\ref{log factor bound amenability} applies whenever the outputs take values in a fixed finite set $A\subset\mathbb R$, with constants depending on $A$.} Restricting to binary profiles allows us to improve the square root bound to a bound that is
linear up to a logarithmic factor, as presented in the following main result. 
\begin{theorem}\label{log factor bound amenability}
For every graph \(G=(N,E)\) and every information skeleton  $\mathfrak S=(N,K,B)$,  
\[
\mathtt I_{\mathrm{leader}}(G,\mathfrak S)\leq 2h\left(\tfrac 1 2\mathtt I_{\mathrm{bin}}(G,\mathfrak S)\right)\, ,
\]
where $h(p):=-p\log_2 p-(1-p)\log_2(1-p)$ denotes the binary entropy function\footnote{With the convention \(0\log 0=0\).}.
\end{theorem}

\begin{remark}
Since
$h(\varepsilon)\le2\varepsilon\log_2(1/\varepsilon)$
for $\varepsilon\in(0,1/2]$, we obtain
\[
\mathtt I_{\mathrm{leader}}(G,\mathfrak S)\leq 2\mathtt I_{\mathrm{bin}}(G,\mathfrak S)\log_2(2/\mathtt I_{\mathrm{bin}}(G,\mathfrak S))\, ,
\]
where the right-hand side is interpreted as zero when
$\mathtt I_{\mathrm{bin}}(G,\mathfrak S)=0$.
Moreover, since $h$ is increasing on $[0,1/2]$, whenever
$\mathtt I_{\mathrm{bin}}(G,\mathfrak S)\le\varepsilon$
with $\varepsilon\in(0,1]$, Theorem~\ref{log factor bound amenability}
gives
\(
\mathtt I_{\mathrm{leader}}(G,\mathfrak S)
\le 2h(\varepsilon/2)
\le 2\varepsilon\log_2(2/\varepsilon).
\)
Thus, $G$ is
$(2\varepsilon\log_2(2/\varepsilon),\mathfrak S)$-amenable.
This improves the square root bound
in~\cite[Theorem~3]{TOOsqrtbound2026}.
\end{remark}

\noindent\textbf{Proof overview.} Our proof of Theorem~\ref{log factor bound amenability} keeps the same general architecture as the proof of Theorem~3 in~\cite{TOOsqrtbound2026}, which proceeds in two main steps. First, given a coordination profile $\pi=(X_i,Z_j)_{i\in N,j\in K}$, they associate to each player \(i\in N\) a monotone transferable utility coalitional game (TU game) $v_i$ depending only on $X_i,Z_{B(i)}$. Specifically,
\begin{equation}\label{eq:vi-definition}
    v_i(T):=\operatorname{Var}(\mathbb E[X_i\mid Z_{T}])\, .
\end{equation} The Shapley value of $v_i$ is a probability measure
\begin{equation}\label{eq:Shapley value measure definition}
\mu_i:=\mathrm{Sh}(v_i)    
\end{equation}
supported on \(B(i)\). They prove a
contraction estimate saying that for any two players 
\begin{equation} \label{TV_bound_sqrt}
    \|\mu _i-\mu_j\|_{\operatorname{TV}} \le c\sqrt{(1-\operatorname{Cov}(X_i,X_j))}\, ,
\end{equation}
with $c=\sqrt{2}$.

Second, they use a grand coupling theorem (see~\cite{AngelSpinka2019}) to couple the measures
\((\mu_i)_{i\in N}\), producing random leaders \(L_i\in B(i)\). Players
with the same leader form communities, and the probability that an edge crosses
between communities is controlled by the total variation distance between the
corresponding probability measures.

\medskip

In the present argument we replace the variance-based games $v_i$ by information-theoretic CDFs 
\[
f_i(j):= I(X_i;Z_1,\ldots,Z_j)\, , 
\]
where $I(X;Y)$ is the mutual information function. Then we set $\mu_i=\mathrm df_i$. The total variation distance between these measures is bounded by the binary entropy of the disagreement probability, giving the improved bound through the same coupling step.$\ \square$

\medskip

Surprisingly, even while restricting ourselves to binary profiles and working with the original distance-based model of HRT, it is not possible to achieve a universal linear bound on $\Ilead$ in terms of $\Ibin$. This is established by the following result, which shows the existence of a family of graphs for which no such universal bound holds due to a polylogarithmic loss.


\begin{theorem}\label{thm:counter to linearity}
Fix an integer $r\ge1$. For every sufficiently large integer $k$
(depending on $r$), there exists a finite graph $G$ such that
\[
\Ilead(G,r)=\Omega\left(\sqrt{k}\,2^{-k}\right)
\quad\text{and}\quad
\Ibin(G,r)=O(2^{-k}),
\]
where the implicit constants are independent of $k$ and $r$.
In particular, there exists a universal constant $C>0$
such that these graphs satisfy
\[
\Ilead(G,r)
\ge C\,\Ibin(G,r)
\sqrt{\log_2\!\left(1/\Ibin(G,r)\right)}.
\]
\end{theorem}


\noindent\textbf{Proof overview.}
For $r=1$, we construct an auxiliary graph $H$ whose vertices are
odd, monotone weighted majority functions on $k$ binary inputs.
Adjacent functions differ on exactly one antipodal pair, so their
disagreement probability is $2^{1-k}$.
We construct $G$ by taking $N$ copies of each function and adding
$k$ special vertices carrying independent fair signs. Each copy
observes the special vertices corresponding to its relevant
coordinates and outputs its function's value. Edges between copies
of adjacent functions are included independently with probability
$p_{fg}$. Taking $N$ large makes the edges incident to special
vertices negligible, giving binary cost $O(2^{-k})$.

For the leader lower bound, we first reduce to assignments in which
every ordinary vertex follows a special vertex, without increasing
the number of cut ordinary edges. We choose $p_{fg}$ using crossing
probabilities in the hyperplane arrangement defining $H$.
An auxiliary inequality then forces every such assignment to cut
$\Omega(N^2)$ edges in expectation, while the expected total number
of ordinary edges is $\Theta(N^2 2^k/\sqrt{k})$.
Concentration makes these bounds hold simultaneously for all
assignments, yielding leader inefficiency
$\Omega(\sqrt{k}\,2^{-k})$.$\ \square$

\medskip

The proof is presented in Section~\ref{sec:counter to linearity}, together with a more detailed overview and short explanations of how it easily extends for any fixed\footnote{In fact, the extension holds for any radius which satisfies $r < 2^{k-2}/k$.} $r$ and why the constructions discussed in the proof can be implemented as equilibria (thus disproving the conjecture of~\cite[Section~7]{TOOsqrtbound2026} that leader equilibria are always optimal up to constant factor). Given the restrictive nature of the construction, which employs distance-$1$ binary profiles defined by weighted majority rules, we believe that there may exist an example demonstrating the sharpness of the bound in Theorem~\ref{log factor bound amenability} in the more general setting of information skeletons.

\medskip

Returning to the original framework of~\cite{TOOsqrtbound2026}, we also establish a logarithmic bound, possibly with worse constants, using the same variance-based TU games and their Shapley value measures. Namely, the total variation distance between two such measures is bounded by the disagreement probability times a logarithmic factor, which directly strengthens~\eqref{TV_bound_sqrt} for binary profiles. The proof of the following proposition is given in Section~\ref{sec:var_game_TV_bound}.

\begin{proposition}\label{prop:varaince games}
For every binary coordination profile $(X_i,Z_j)_{i\in N,j\in K}$, and any two players $i,j\in N$ with $\varepsilon_{ij}:=1-\operatorname{Cov}(X_i,X_j)\leq \frac 1 2$, the Shapley value measures $\mu_i,\mu_j$ of the variance games satisfy
\[
\|\mu_i-\mu_j\|_{\operatorname{TV}} \le C\,\varepsilon_{ij}\log(1/\varepsilon_{ij})\, ,
\]
for some universal constant $C>0$.
\end{proposition}

The following result illustrates a further subtlety: while the bound for the Shapley measures in Proposition~\ref{prop:varaince games} contains a logarithmic factor, the Hausdorff distance between the full cores of two variance games is bounded linearly by the disagreement probability.

\begin{proposition}\label{prop:coreH dist}
Let \(X,Y\) be unbiased sign random variables measurable with respect to independent \((Z_j)_{j\in K}\). Then
\[
d^{\operatorname{H}}_{\operatorname{TV}} \bigl(\operatorname{Core}(v_X), \operatorname{Core}(v_Y) \bigr) \leq  4\mathbb P(X\neq Y).
\]
Here
\(
\operatorname{Core}(v_X)
:=
\left\{
\mu\in\Delta(K):
\mu(S)\geq v_X(S)\text{ for every }S\subseteq K
\right\},
\) and \(d^{\operatorname{H}}_{\operatorname{TV}}\) denotes the Hausdorff distance induced by the total variation metric.
\end{proposition}

\begin{remark}
    In particular, Proposition~\ref{prop:coreH dist} suggests that a linear bound holds in cases where one can establish a simultaneous selection principle for these cores. More precisely, suppose that, in certain special cases, there exists a universal constant $C>0$ such that for every binary coordination profile
\(\pi=(X_i,Z_j)_{i\in N,j\in K}\),
\[
\min_{\mu_i\in\operatorname{Core}(v_{X_i})}
\sum_{\{i,j\}\in E}
\|\mu_i-\mu_j\|_{\operatorname{TV}}
\leq
C\sum_{\{i,j\}\in E}
d^{\operatorname{H}}_{\operatorname{TV}}
\bigl(
\operatorname{Core}(v_{X_i}),
\operatorname{Core}(v_{X_j})
\bigr).
\]
Then Proposition~\ref{prop:coreH dist} would bound the right-hand side
linearly by the total disagreement probability. Moreover, since
\(X_i\) is measurable with respect to \(Z_{B(i)}\), every
\(\mu_i\in\operatorname{Core}(v_{X_i})\) is supported on \(B(i)\).
The grand coupling theorem could therefore be applied to the selected
measures exactly as in the proof of Theorem~\ref{log factor bound amenability},
yielding
\[
\mathtt I_{\mathrm{leader}}(G,\mathfrak S)
\leq C'\mathtt I_{\mathrm{bin}}(G,\mathfrak S)
\]
for a universal constant \(C'>0\). 

\end{remark}

For trees for example, the simultaneous selection can be carried out recursively,
giving the following linear bound.

\begin{corollary}\label{cor:tree-linear-bound}
For every finite tree $T=(N,E)$ and information skeleton
$\mathfrak S=(N,K,B)$,
\[
\mathtt I_{\mathrm{leader}}(T,\mathfrak S)
\le 4\,\mathtt I_{\mathrm{bin}}(T,\mathfrak S).
\]
\end{corollary}

Lastly, we show for any fixed radius and bounded degree the relation between inefficiency and amenability is linear.
\begin{proposition}
\label{prop:C(r,Delta)}
For all integers $r,\Delta>0$ there exists $C(r,\Delta)>0$ such that for every graph of maximal degree at most $\Delta$, we have
\[
\mathtt I_{\mathrm{leader}}(G,r)\leq C(r,\Delta)\mathtt I(G,r)\, .
\]
\end{proposition}
The proof of Proposition~\ref{prop:C(r,Delta)} is in Section~\ref{sec:additional proofs}.

\section{Proof of Theorem~\ref{log factor bound amenability}}

We first recall the relevant information-theoretic notation (for a detailed exposition of the topic, we refer the reader to~\cite{InformationTheoryBook},\cite{InfoOnNetworks}). Let \(X,Y,Z\) be random variables. We write $H$ for Shannon entropy, measured in bits using base-$2$
logarithms, so that $H(X)=1$ for every unbiased binary variable $X$. The
mutual information of $X,Y$ is defined as
\[
I(X;Y):=H(X)-H(X\mid Y).
\]
The conditional mutual information of $X,Y$ conditioned on $Z$ is defined as
\[
I(X;Y\mid Z):=H(X\mid Z)-H(X\mid Y,Z).
\]
We also recall the chain rule
\[
I((X,Y);Z)
=
I(X;Z)+I(Y;Z\mid X),
\]
and the non-negativity property
\[
I(X;Y\mid Z)\ge 0.
\]

Let $G=(N,E)$ be a finite graph and $\mathfrak S=(N,K,B)$ be an information skeleton. We fix an ordering on $K = \{1,\dots,k\}$. Let $\pi=(X_i,Z_j)_{i\in N,j\in K}$ be a binary coordination profile, and denote $Z_{[j]}:= \{Z_1,\dots,Z_j\}$. 
For a finite-valued random variable $X$ measurable w.r.t.\ $Z_1,\ldots,Z_k$ define a function 

\[
f_X\colon\{0,\dots,k\} \to [0,H(X)], \quad f_X(j):= I(X;Z_{[j]}).
\]
Notice that $f_X$ is monotone, with \(f_X(0)=0\), \(f_X(k)=H(X)\), so $\mathrm df_X$ is a measure on $K$ defined by 
\[
\mathrm df_X(\{j\})= f_X(j)-f_X(j-1)\, .
\]
For $i\in N$, let $\mu_i:=\mathrm{d}f_{X_i}$, which is a probability measure with $\operatorname{supp}(\mu_i)\subseteq B(i)$.

\medskip

We recall that the total variation norm of a signed measure $\mu$ on a set $K$ is given by

\[
\|\mu\|_{\operatorname{TV}}
= \frac12
\sum_{j\in K} \bigl|\mu(\{j\})\bigr|\, .
\]
In the following Lemma, we show that the total variation distance of the probability measures is bounded by the binary entropy function.

\begin{lemma}\label{lem:ordered-information-tv}
Let \(X,Y\) be unbiased sign random variables measurable with respect to \(Z_1,\dots,Z_k\) with $p=\mathbb{P}(X\neq Y)$. Then 
\[
\|\mathrm d f_X-\mathrm d f_Y\|_{\operatorname{TV}}
\le
h(p).
\]
\end{lemma}
\begin{proof}
Consider the random variable $(X,Y)$. By the chain rule,
\[
I(X,Y; Z_{[j]})
=
I(X;Z_{[j]})+I(Y;Z_{[j]}\mid X).
\]
Hence the function
\(
(f_{(X,Y)}-f_{X})\,(j)
=
I(Y;Z_{[j]}\mid X)
\) is monotone and, therefore, 
\[
\|\mathrm d f_{(X,Y)}-\mathrm d f_X\|_{\operatorname{TV}}=\tfrac 1 2 (f_{(X,Y)}- f_X)\,(k) =\tfrac 1 2H(Y\mid X).
\]
Similarly,
\[
\|\mathrm d f_{(X,Y)}-\mathrm d f_Y\|_{\operatorname{TV}}=\tfrac 1 2H(X\mid Y).
\]
It follows that 
\[
\| \mathrm d f_X-\mathrm d f_Y\|_{\operatorname{TV}}\leq \|\mathrm d f_{(X,Y)}-\mathrm d f_X\|_{\operatorname{TV}}+\|\mathrm d f_{(X,Y)}-\mathrm d f_Y\|_{\operatorname{TV}}=\tfrac 1 2(H(Y\mid X)+H(X\mid Y)).
\]

It remains to compute the two conditional entropies. Recall \(X\) and \(Y\) are both unbiased sign variables and \(p=\mathbb{P}[X\neq Y]\), thus, conditional on \(Y\), the variable \(X\) differs from \(Y\) with
probability \(p\). Hence
\[
H(X\mid Y)=h(p), \quad \text{and similarly} \quad H(Y\mid X)=h(p).
\]
Therefore
\[
\| \mathrm d f_X-\mathrm d f_Y\|_{\operatorname{TV}}\le
h(p)\, ,
\]
as wanted.
\end{proof}

\medskip 
Using the Lemma above we are ready to prove Theorem~\ref{log factor bound amenability}---this part follows the same route as in the second step in the proof of~\cite[Theorem~3]{TOOsqrtbound2026}.

\medskip

\begin{proof}[Proof of Theorem~\ref{log factor bound amenability}]
For every edge \(\{i,j\}\in E\), set
\[
p_{ij}:=\mathbb{P}[X_i\neq X_j].
\]
Applying Lemma~\ref{lem:ordered-information-tv} to \(X_i\) and \(X_j\), we get
\[
\|\mu_i - \mu_j\|_{\operatorname{TV}}\le h(p_{ij}).
\]

As in~\cite{TOOsqrtbound2026}, we use the following coupling theorem.
\begin{theorem}[\cite{KleinbergTardos2002,AngelSpinka2019}]\label{coupling}
Let \((\mu_i)_{i\in N}\) be a finite or countable family of probability measures on a common
finite set \(K\). Then there exists a coupling of random variables $(L_i)_{i\in N}$ such that $L_i\sim \mu_i$ for every \(i\), and for every pair \(i,j\in N\),
\[
\mathbb{P}[L_i\neq L_j]
\le
\frac{2\|\mu_i - \mu_j\|_{\operatorname{TV}}}
{1+\|\mu_i - \mu_j\|_{\operatorname{TV}}}
\le
2\|\mu_i - \mu_j\|_{\operatorname{TV}}.
\]
\end{theorem}

By Theorem~\ref{coupling} and since $\operatorname{supp}(\mu_i)\subseteq B(i)$, there exist random variables $(L_i)_{i\in N}$ such that \(L_i\in  B(i)\) for every \(i\in N\), and for every edge \(\{i,j\}\in E\),
\[
\mathbb{P}[L_i\neq L_j]\le 2h(p_{ij}).
\]

Averaging over edges gives
\[
\frac1{|E|}
\sum_{\{i,j\}\in E}
\mathbb{P}[L_i\neq L_j]
\le
\frac2{|E|}
\sum_{\{i,j\}\in E}h(p_{ij}).
\]
The binary entropy function \(h\) is concave on \([0,1]\). Hence Jensen's inequality gives
\[
\frac2{|E|}
\sum_{\{i,j\}\in E}h(p_{ij})
\le
2h\left(
\frac1{|E|}
\sum_{\{i,j\}\in E}p_{ij}
\right).
\]
Since
\[
\frac1{|E|}
\sum_{\{i,j\}\in E}p_{ij}=\frac 1 2\, \mathtt C(G,\pi)\, ,
\]
we obtain
\[
\frac1{|E|}
\sum_{\{i,j\}\in E}
\mathbb{P}[L_i\neq L_j]
\le
2h\left(\tfrac 1 2\, \mathtt C(G,\pi)\right).
\]
Take a realization of $\{L_i\}_{i\in N}$, such that
\[
\frac1{|E|}
\sum_{\{i,j\}\in E}
\mathbf 1_{\{L_i\neq L_j\}}
\le
2h\left(\tfrac 1 2\, \mathtt C(G,\pi)\right).
\]
This defines a leader coordination profile $\pi_{\mathrm{leader}}=(X_i^{\mathrm{leader}},Z_j^{\mathrm{leader}})_{i\in N,j\in K}$ where \(Z^{\mathrm{leader}}_{j}\) are independent unbiased signs, and \(X^{\mathrm{leader}}_i = Z^{\mathrm{leader}}_{L_i}\), so we conclude 
\[
\mathtt I_{\mathrm{leader}}(G,\mathfrak S)\leq 2h\left(\tfrac 1 2\, \mathtt C(G,\pi)\right).
\]
This holds for every binary coordination profile $\pi$, then by continuity of $h$,
\[
\mathtt I_{\mathrm{leader}}(G,\mathfrak S)\leq 2h\left(\tfrac 12\,\mathtt I_{\mathrm{bin}}(G,\mathfrak S)\,\right).
\]
\end{proof}

\begin{remark}\label{remark: info game bound is tight}
    The improvement of Theorem~\ref{log factor bound amenability} lies in the improved bound established in Lemma~\ref{lem:ordered-information-tv}. We remark that this bound is sharp. Indeed, fix some $p\in (0,1]$ and let $Z_1$ be an unbiased sign random variable, and let $Z_2$ be an independent biased sign random variable taking the value $1$ with probability $p$. Set $X = Z_1 , Y = -Z_1 \cdot Z_2$. Notice that 
  
\[
    f_{X}(j) = \begin{cases}
    0,&j=0\\
    1,&j=1\\
    1,&j=2
    \end{cases} \qquad 
    f_{Y}(j) = \begin{cases}
    0,&j=0\\
    1-h(p),&j=1\\
    1,&j=2
    \end{cases}
\]
Thus, one may verify that $\mathrm d f_X= (1,0) , \mathrm d f_Y = (1-h(p), h(p))$, and $\|\mathrm d f_X-\mathrm d f_Y\|_{\operatorname{TV}} = h(p)$. This shows the sharpness, as $\mathbb{P}(X\neq Y)=p$.
\end{remark}

\section{Proof of Theorem~\ref{thm:counter to linearity}} \label{sec:counter to linearity}
We first prove the theorem for $r=1$. The proof is divided into
several steps, which we describe here. We begin by defining the
auxiliary graph $H$, whose vertices represent odd, monotone weighted
majority functions on $k$ binary inputs (see~\cite{BooleanFunctions}
for background on Boolean functions). Two vertices are connected
if their functions agree on every input except exactly one
antipodal pair $\{x,-x\}$. Thus adjacent functions disagree with
probability $2^{1-k}$ under a uniformly random input. This graph
also has a natural geometric interpretation: the hyperplanes
$w\cdot x=0$, for $x\in\{\pm1\}^k$, partition the positive orthant
into regions corresponding to weighted majority functions, and
adjacent functions correspond to regions sharing a facet.

We then construct a random graph $G$, using the same notation for
its realizations. We add $k$ special vertices representing the
input coordinates and $N$ copies of every vertex of $H$.
We call these copies ordinary vertices. Each copy of a function
$f$ is connected to the special vertices corresponding to the
coordinates on which $f$ depends, whose set we denote by $\supp(f)$.
For each pair of adjacent functions $f,g$ in $H$, we also include
each of the $N^2$ possible edges between their copies independently
with probability $p_{fg}$, which we define later.
For fixed $k$ and these probabilities, taking $N$ sufficiently
large makes the contribution of edges incident to special vertices
negligible: their number grows linearly with $N$, whereas the
expected number of ordinary edges grows quadratically.

This construction gives a natural binary profile. Each special
vertex outputs its independent fair sign, and each ordinary vertex
outputs the value of its function on the observed signs.
This is admissible because every coordinate relevant to the
function is visible to that vertex. Along every ordinary edge,
the disagreement probability is $2^{1-k}$, so, once the contribution
of the special edges is negligible, the profile has cost $O(2^{-k})$.

The main difficulty is to bound from below the number of edges
whose endpoints choose different leaders, uniformly over all
admissible leader assignments. We first show that every assignment
can be replaced by one in which all ordinary vertices follow special
vertices, without increasing the number of cut ordinary edges.
It therefore suffices to consider these assignments.

The probabilities $p_{fg}$ are chosen geometrically, using the
probabilities of crossing the common facets between regions when
weights are sampled uniformly from the simplex and one weight is
decreased to zero. This choice serves two purposes. First, an auxiliary
inequality gives an expected lower bound of order $N^2$ on the
number of cut ordinary edges for every fixed assignment to special
vertices. Second, a calculation of $\sum_{\{f,g\}\in E(H)}p_{fg}$
shows that the expected total number of ordinary edges is of order
$N^2 2^k/\sqrt{k}$. Concentration estimates and a union bound ensure
that, for sufficiently large $N$, a single realization satisfies
the cut lower bound simultaneously for all assignments, together
with the required edge-count estimate. Dividing these bounds yields
leader inefficiency $\Omega(\sqrt{k}\,2^{-k})$, establishing the
separation from the binary profile.

\medskip

\paragraph{Step~1 (Auxiliary graph construction):}~\\
\\
Fix an integer $k\ge5$. A weighted majority function determined
by weights $w=(w_1,\ldots,w_k)$ is a binary function
$f_w:\{\pm1\}^k\to\{\pm1\}$ defined by
\[
f_w(x)=\operatorname{sign}(w_1x_1+\dots+w_kx_k).
\]
We consider only positive weights for which $w\cdot x\neq0$ for every
$x\in\{\pm1\}^k$. We want to bound $\Prob(f(x)\neq g(x))$ for two
majority functions $f,g$ and $x\in\{\pm1\}^k$ chosen uniformly at random.

For each $x\in\{\pm1\}^k$, consider the hyperplane
$x^\perp=\{w\in[0,\infty)^k:x\cdot w=0\}$.
On one side of this hyperplane $f_w(x)=1$, and on the other
$f_w(x)=-1$. These hyperplanes divide the positive orthant into
open convex cones, and all weights in the same region determine the
same majority function.

We define our auxiliary graph $H$ as follows:
\[
V(H)=\{\text{weighted majority functions as above}\},
\]
\[
E(H)=\{\{f,g\}:\text{the regions of }f,g\text{ share a facet
inside }(0,\infty)^k\}.
\]
Each vertex thus corresponds to one region. Notice that $f\sim g$ in
$H$ if and only if there is exactly one antipodal pair $\{x,-x\}$ on
which $f$ and $g$ differ. See Figure~\ref{fig:H}.
\begin{figure}
    \centering
    \includegraphics[width=0.35\linewidth]{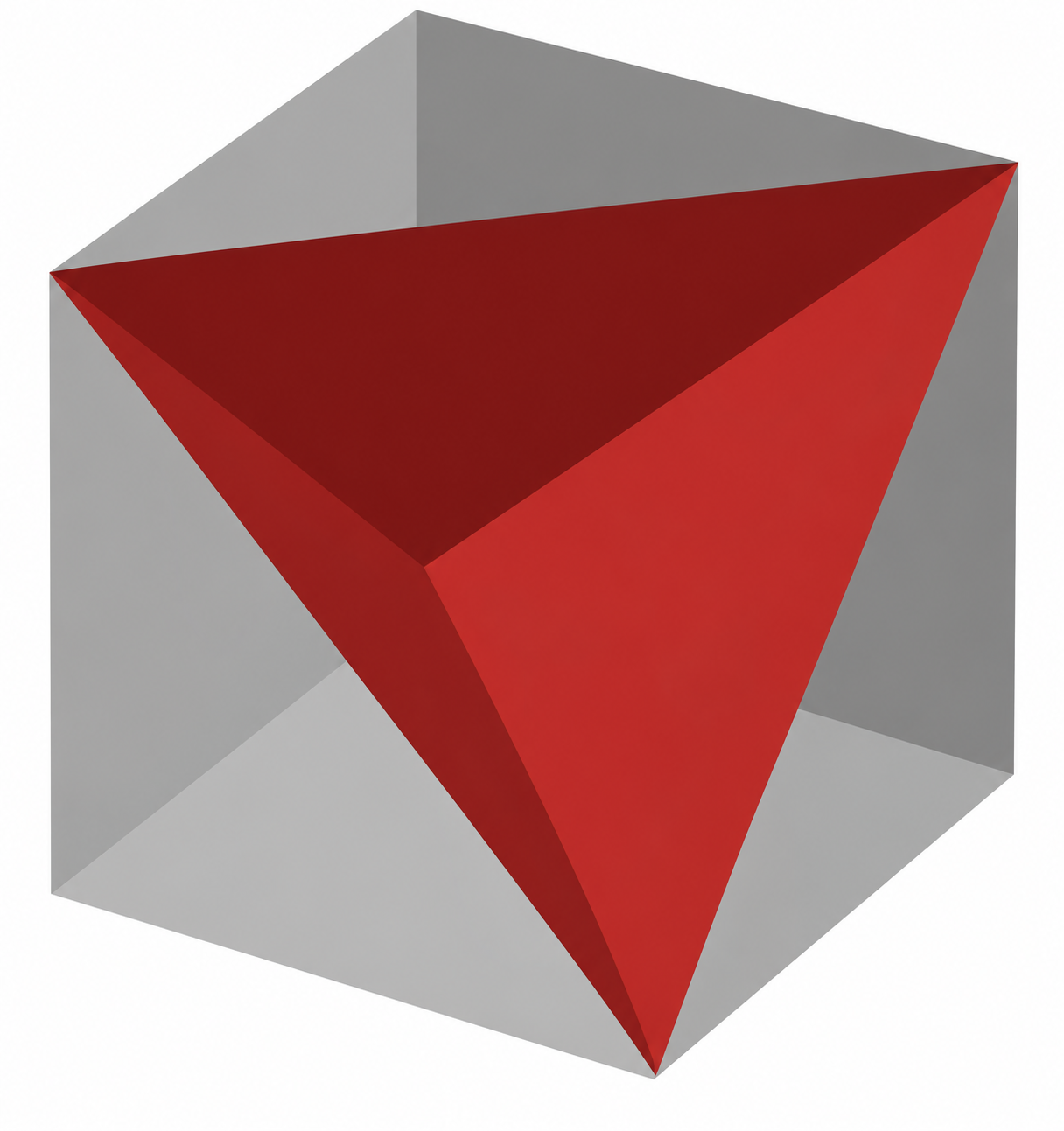}
    \caption{While the formal proof requires \(k \ge 5\), we consider \(k=3\) here for visualization purposes. The parts of the hyperplanes defining $H$ are shown in red. They partition the positive orthant with its coordinate boundaries $\{w_1=0\},\{w_2=0\},\{w_3=0\}$ shown in grey (all of which an imaginative reader is invited to visualize as extending out of the page), into regions corresponding to the vertices of $H$ (which in this case is a star graph with four vertices).
    }
    \label{fig:H}
\end{figure}

In particular, for $f\sim g$ we have
$\Prob(f\neq g)=2/2^k=1/2^{k-1}$, and in general, by the union bound,
\[
\Prob(f\neq g)\le \operatorname{dist}_H(f,g)/2^{k-1}.
\]

\paragraph{Step~2 (Defining the random graph):}~\\
\\
For each majority function $f$, denote
\[
\supp(f)=\{i\in[k]:\exists x\in\{\pm1\}^k
\text{ such that }f(x_{-i},1)\neq f(x_{-i},-1)\}.
\]
This is the set of indices that affect the value of the function for
at least one input. Let $N\in\mathbb N$ and let $p_{fg}\in[0,1]$ for
every $f\sim g$ in $H$. Define the graph $G$ by
\[
V(G)=\bigcup_{f\in V(H)}V_f\cup\{v_1,\ldots,v_k\},
\]
where $V_f=\{u_f^1,\ldots,u_f^N\}$ consists of $N$ distinct vertices
corresponding to $f$.

$E(G)$ contains all edges $\{v_i,u_f^j\}$ for $1\le j\le N$ and
$i\in\supp(f)$. For each $f\sim g$, it also contains, independently,
each of the $N^2$ possible edges between $V_f$ and $V_g$ with
probability $p_{fg}$, and no other edges.
Taking $N$ large makes these edges dominate those incident to the
special vertices, while $p_{fg}$ controls how much each adjacent pair
contributes to the cost. We will show that, for suitable $p_{fg}$ and
sufficiently large $N$, a realization of $G$ has the desired property. See Figure~\ref{fig:G}.

\begin{figure}
    \centering
    \includegraphics[width=0.8\linewidth]{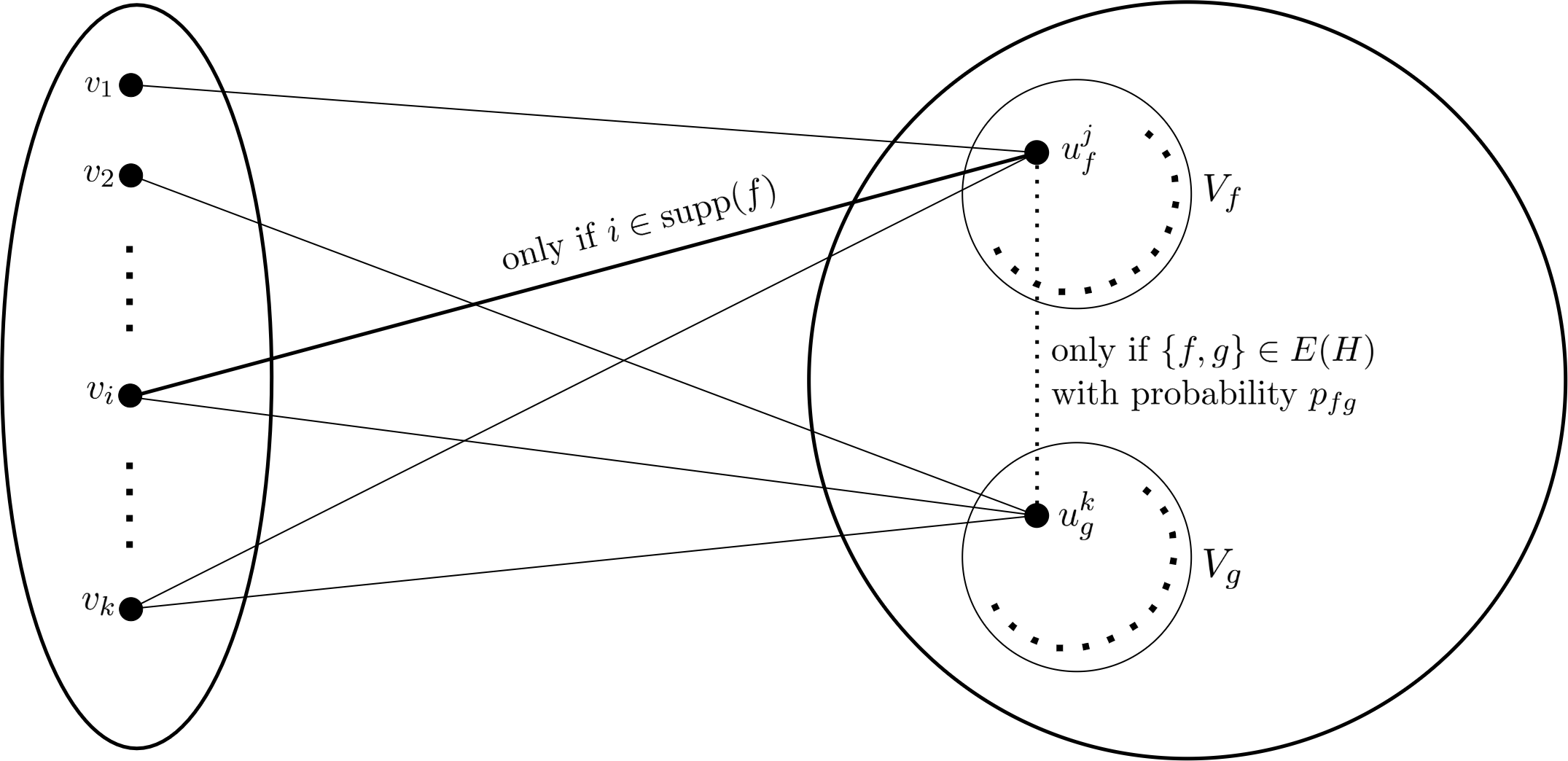}
    \caption{A sketch of a realization of the random graph $G$. On the left are the special vertices, representing the input coordinates of the weighted majority functions, and on the right are the ordinary vertices, with $N$ copies of each vertex of $H$.}
    \label{fig:G}
\end{figure}

\paragraph{Step~3 (Upper bound for the binary coordination profile):}~\\
\\
Assign each vertex $v\in V(G)$ an independent fair
$\{\pm1\}$-valued random variable $Z_v$. We give an explicit binary
profile:
\[
X_{v_i}=Z_{v_i},
\qquad
X_{u_f^j}=f(\{Z_{v_i}\}_{i\in\supp(f)})
=f(Z_{v_1},\ldots,Z_{v_k}).
\]
This slightly abuses notation, but since $f$ depends only on
coordinates in $\supp(f)$, the profile is well defined and uses only
radius-one information. Each output is unbiased, since $f$ is odd.
Denote
\[
a=\sum_{f\in V(H)}|\supp(f)|,
\qquad
A=\#\{\text{edges of }G\text{ between sets }V_f,V_g\}.
\]
Notice that $|E(G)|=A+Na$. On every edge between $V_f$ and $V_g$,
\[
\Prob\bigl(f(Z_{v_1},\ldots,Z_{v_k})
\neq g(Z_{v_1},\ldots,Z_{v_k})\bigr)=1/2^{k-1}.
\]
The binary cost is twice the average disagreement probability, and
monotonicity gives disagreement probability at most $1/2$ on edges
incident to special vertices. Thus
\[
\texttt{C}(G,X)
\le\frac{(2/2^{k-1})A+Na}{A+Na}
\le 2/2^{k-1}+\frac{Na}{A}.
\]
Since $A$ is a sum of independent Bernoulli variables with mean
$N^2\sum_{\{f,g\}\in E(H)}p_{fg}$, with probability tending to $1$ as
$N\to\infty$,
\[
A\ge\frac{N^2}{2}\sum_{\{f,g\}\in E(H)}p_{fg}.
\]
For each fixed $k$, taking $N$ sufficiently large therefore gives,
with probability tending to $1$,
\[
\Ibin(G,1)\le\texttt{C}(G,X)=O(2^{-k}).
\]

\paragraph{Step~4 (Lower bound for leader strategies):}~\\
\\
This step is divided into several sub-steps. We first define the probabilities $p_{fg}$. The definition relies on the geometry of the hyperplane arrangement which defines $H$, and it gives the needed framework for proving (in the next sub-step) Lemmas~\ref{lem:key-inequality} and~\ref{lem:sum-probabilities}, which lie at the heart of the argument which gives the desired lower bound for leader strategies. The following sub-step is to show that without loss of generality, any assignment of leaders can be assumed to assign \say{leadership} only to special vertices. Finally, we combine all the lemmas to establish the desired lower bound, thereby completing the proof.

\paragraph{Step~4.1 (Choice of the probabilities):}~\\
\\
We now choose the probabilities $p_{fg}$. Since multiplying all weights
by the same positive number does not change the majority function,
we consider weights whose sum is $1$. Let
\[
\Delta=\left\{w\in[0,\infty)^k:\sum_{i=1}^k w_i=1\right\},
\]
and choose $W=(W_1,\ldots,W_k)$ uniformly from $\Delta$.
For each $i\in[k]$, consider the path
\[
w^{(i)}(t)=(W_1,\ldots,W_{i-1},t,W_{i+1},\ldots,W_k),
\qquad 0<t\le W_i.
\]
That is, we keep all weights except the $i$th fixed and increase the
$i$th weight from $0$ to $W_i$. The path need not lie in $\Delta$.
For $f\sim g$, let $F_{fg}$ be the common facet inside $(0,\infty)^k$.
We define
\[
p_{fg}:=\Prob\bigl(w^{(i)}\text{ crosses }F_{fg}\bigr).
\]
This probability does not depend on $i$. To see this, let
\[
T_{fg}=\left\{w\in F_{fg}:\sum_{j=1}^k w_j<1\right\}.
\]
Deleting coordinate $i$ from $W$ gives a uniform point in the solid
simplex of nonnegative coordinates with sum at most $1$, whose
volume is $1/(k-1)!$ (this is a classical result, see~\cite{simplexvol} for example). The path crosses $F_{fg}$ exactly when this point lies in
the projection of $T_{fg}$. Its crossing probability is therefore
the projected volume divided by $1/(k-1)!$. Since each face lies in a hyperplane $x^{\perp}$ for some $x \in \{\pm 1\}^k$, the unit normal
to the facet has all coordinates of absolute value $1/\sqrt{k}$,
so deleting any coordinate multiplies its surface volume by
$1/\sqrt{k}$ (this is Cauchy's projection formula, see~\cite{projectionformula}). All these projections consequently have the same
volume, so $p_{fg}$ is independent of $i$ (and positive).

\paragraph{Step~4.2 (Auxiliary inequality and sum of probabilities):}

\begin{lemma}\label{lem:key-inequality}
For each $f\in V(H)$, let $\mu_f$ be a probability measure on
$\supp(f)$. Then
\[
\sum_{\{f,g\}\in E(H)}p_{fg}\,\|\mu_f-\mu_g\|_{\operatorname{TV}}\ge\frac12.
\]
\end{lemma}

\begin{proof}
For each unordered edge $\{f,g\}$, fix an ordering $(f,g)$ and
choose $x_{fg}\in\{\pm1\}^k$ such that
$f(x_{fg})=1$ and $g(x_{fg})=-1$.

Fix $i\in[k]$ and follow the path $w^{(i)}$. All signed sums
$\sum_{j\neq i}W_jx_j$ are nonzero almost surely, so for all
sufficiently small $t>0$ the output of $f_{w^{(i)}(t)}$ is
independent of its $i$th input. Therefore
\[
i\notin\supp\bigl(f_{w^{(i)}(t)}\bigr),
\qquad
\mu_{f_{w^{(i)}(t)}}(i)=0
\quad\text{for all sufficiently small }t>0.
\]
At the end of the path, this quantity equals $\mu_{f_W}(i)$.

Along the path, $\mu_{f_{w^{(i)}(t)}}(i)$ is constant until we cross
a facet. If the path crosses $F_{fg}$, the change is
$x_{fg,i}(\mu_f(i)-\mu_g(i))$, since increasing $w_i$ crosses from
$g$ to $f$ when $x_{fg,i}=1$, and from $f$ to $g$ when $x_{fg,i}=-1$.
Summing the changes along the path and taking expectations gives
\[
\mathbb E\bigl[\mu_{f_W}(i)\bigr]
=\sum_{\{f,g\}\in E(H)}
p_{fg}x_{fg,i}\bigl(\mu_f(i)-\mu_g(i)\bigr).
\]
Summing over $i\in[k]$ yields
\[
\begin{aligned}
1
&=\sum_{\{f,g\}\in E(H)}p_{fg}
\sum_{i=1}^k x_{fg,i}\bigl(\mu_f(i)-\mu_g(i)\bigr)\\
&\le\sum_{\{f,g\}\in E(H)}p_{fg}
\sum_{i=1}^k|\mu_f(i)-\mu_g(i)|\\
&=2\sum_{\{f,g\}\in E(H)}p_{fg}\,\|\mu_f-\mu_g\|_{\operatorname{TV}}.
\end{aligned}
\]
\end{proof}

We next evaluate the sum of the probabilities.

\begin{lemma}\label{lem:sum-probabilities}
The above choice satisfies
\[
P:=\sum_{\{f,g\}\in E(H)}p_{fg}
=2^{-k}\binom{2k-2}{k-1}
=\Theta\!\left(\frac{2^k}{\sqrt{k}}\right).
\]
\end{lemma}

\begin{proof}
We group the facets according to the hyperplane containing them.
Fix $x\in\{\pm1\}^k$ with exactly $m$ positive coordinates, where
$1\le m\le k-1$. A path crosses $x^\perp$ at most once, and when it
does, it crosses exactly one of its facets, almost surely. Thus the
sum of $p_{fg}$ over these facets is the probability of crossing
$x^\perp$.

Since the crossing probabilities are independent of $i$, choose
$i$ with $x_i=1$ and write
\[
a=\sum_{\substack{j\neq i\\x_j=1}}W_j,
\qquad
b=\sum_{x_j=-1}W_j.
\]
Along the $i$th path the signed sum is $t+a-b$, which vanishes at
$t=b-a$. Since $W_i=1-a-b$, the path crosses $x^\perp$ exactly when
\[
0<b-a<1-a-b,
\qquad\text{equivalently}\qquad a<b<\frac12.
\]

We calculate the volume of this region by integrating over $b=s$.
For each $s$, the $m-1$ positive-side coordinates other than $i$
must have sum less than $s$, giving volume $s^{m-1}/(m-1)!$ (again by~\cite{simplexvol}).
The $k-m$ negative-side coordinates with sum at most $s$ have
volume $s^{k-m}/(k-m)!$. Differentiating gives their volume element
at sum $s$:
\[
\frac{d}{ds}\left(\frac{s^{k-m}}{(k-m)!}\right)\,ds
=\frac{s^{k-m-1}}{(k-m-1)!}\,ds.
\]
For $m=1$, the first factor is $1$, since there are no other
positive-side coordinates. Also, $a<b<1/2$ already implies
$a+b<1$, so the simplex imposes no further restriction.
Multiplying by the density $(k-1)!$ of $(W_j)_{j\neq i}$, the
crossing probability is
\[
\begin{aligned}
(k-1)!\int_0^{1/2}
\frac{s^{m-1}}{(m-1)!}
\frac{s^{k-m-1}}{(k-m-1)!}\,ds
&=2^{1-k}\binom{k-2}{m-1}.
\end{aligned}
\]
There are $\binom{k}{m}$ choices of $x$ with $m$ positive
coordinates, and $x$ and $-x$ determine the same hyperplane. Hence
\[
\begin{aligned}
P
&=\frac12\sum_{m=1}^{k-1}\binom{k}{m}
2^{1-k}\binom{k-2}{m-1}\\
&=2^{-k}\binom{2k-2}{k-1}.
\end{aligned}
\]
For the last equality, write
$\binom{k-2}{m-1}=\binom{k-2}{k-1-m}$ and count choices of $k-1$
elements from two sets of sizes $k$ and $k-2$, according to how
many are chosen from the first. The asymptotic estimate follows from
$\binom{2k-2}{k-1}\sim4^{k-1}/\sqrt{\pi(k-1)}$.
\end{proof}

\paragraph{Step~4.3 (Reducing leader assignments to coordinate assignments):}~\\
\\
By the equivalence established in Section~\ref{sec:theModel}, it suffices to consider
a leader assignment: for each vertex $v\in V(G)$, its leader is
a vertex $\ell(v)\in B_1(v)$. Leaders may be special vertices or
ordinary vertices. We first show how to handle the latter possibility.

\begin{lemma}\label{lem:common-essential-coordinate}
For every $f\in V(H)$ there exists $i_f\in\supp(f)$ such that
$i_f\in\supp(g)$ for every $g\sim f$.
\end{lemma}

\begin{proof}
Since $f$ is odd, it is balanced, and the Poincar\'e inequality
\cite[Section~2.3]{BooleanFunctions} gives
\[
\sum_i\Inf_i(f)\ge \Var(f) = 1.
\]
In particular, there is $i_f\in\supp(f)$ with
$\Inf_{i_f}(f)\ge1/k$. Since in general
$|\Inf_i(f)-\Inf_i(g)|\le2\Prob(f\neq g)$, for every $f\sim g$,
\[
\Inf_{i_f}(g)\ge1/k-1/2^{k-2}>0,
\]
where the strict inequality holds since $k\ge5$. Thus
$i_f\in\supp(g)$.
\end{proof}

\medskip

Now associate a coordinate with every possible leader by the rule
\[
v_i\to i,
\qquad
u_f^j\to i_f.
\]
For each ordinary vertex $u$, let $c(u)$ be the coordinate
associated with its chosen leader $\ell(u)$. Notice that
$c(u_f^j)\in\supp(f)$. If its leader is $v_i$, this follows from
the definition of the graph. If its leader is ordinary, it is
either $u_f^j$ itself, when the claim is immediate, or a vertex
in some $V_g$ with $g\sim f$, in which case
Lemma~\ref{lem:common-essential-coordinate} gives
$c(u_f^j)=i_g\in\supp(f)$.

Moreover, vertices with the same leader receive the same coordinate.
Hence
\[
\mathbf1_{\{\ell(u)\neq\ell(v)\}}
\ge\mathbf1_{\{c(u)\neq c(v)\}}
\]
for every edge between ordinary vertices.

\paragraph{Step~4.4 (A uniform lower bound):}~\\
\\
For any fixed assignment $c$ on the ordinary vertices satisfying
$c(u_f^j)\in\supp(f)$, write $D(c)$ for the number of edges between
ordinary vertices whose endpoints receive different coordinates.
Let
\[
\mu_f(i)=\frac1N\,|\{u\in V_f:c(u)=i\}|.
\]
These are probability measures supported on $\supp(f)$. Between
$V_f$ and $V_g$, exactly
$N^2(1-\sum_i\mu_f(i)\mu_g(i))$ potential edges have differently
assigned endpoints. Consequently,
\[
\begin{aligned}
\mathbb E[D(c)]
&=N^2\sum_{\{f,g\}\in E(H)}p_{fg}
\left(1-\sum_{i=1}^k\mu_f(i)\mu_g(i)\right)\\
&\ge N^2\sum_{\{f,g\}\in E(H)}p_{fg}
\left(1-\sum_{i=1}^k\min(\mu_f(i),\mu_g(i))\right)\\
&=N^2\sum_{\{f,g\}\in E(H)}p_{fg}\,
\|\mu_f-\mu_g\|_{\operatorname{TV}}
\ge N^2/2,
\end{aligned}
\]
where the last inequality is Lemma~\ref{lem:key-inequality}.
Since $D(c)$ is a sum of independent Bernoulli variables, the
Chernoff lower-tail bound gives
\[
\Prob_G\left(D(c)<\frac{N^2}{4}\right)
\le\exp\left(-\frac{N^2}{16}\right).
\]
There are at most $k^{N|V(H)|}$ such assignments $c$. Taking a
union bound over all of them, we obtain
\[
\Prob_G\left(\text{some }c\text{ has }D(c)<\frac{N^2}{4}\right)
\le\exp\left(N|V(H)|\log k-\frac{N^2}{16}\right),
\]
which tends to $0$ for each fixed $k$ as $N\to\infty$.
This holds simultaneously for all such assignments, including
those obtained from leaders chosen after the graph is realized.
Thus, with probability tending to $1$, every leader assignment
cuts at least $N^2/4$ edges between ordinary vertices.

As in the binary profile argument, with probability tending to one, \(A=O(N^2P)\) and therefore also
$|E(G)|=A+Na=O(N^2P)$. For each $k$ and sufficiently large $N$,
there is therefore a realization satisfying both
\[
\Ilead(G,1)
=\Omega\!\left(\frac{N^2/4}{N^2P}\right)
=\Omega(1/P)
=\Omega\!\left(\sqrt{k}\,2^{-k}\right),
\]
by Lemma~\ref{lem:sum-probabilities}, and
\[
\Ibin(G,1)=O(2^{-k}),
\]
with all constants independent of $k$. $\blacksquare$

\paragraph{Extension to arbitrary $r$:}~\\
\\
The construction extends by subdividing edges incident to special vertices, i.e. replacing
each edge to a special vertex by a unique path of
length $r$. The original binary profile (with every subdivision vertex outputting the sign of the special vertex at the end of its path) remains admissible and the added edges contribute negligibly for sufficiently
large $N$. The only change needed in the lower-bound argument is
to strengthen Lemma~\ref{lem:common-essential-coordinate}: each function $f$ has a coordinate $i_f \in \supp(f)$ which
belongs to $\supp(g)$ for every $g$ at distance at most $r$ from $f$
in $H$. Indeed, choose $i$ with $\Inf_{i}(f)\ge1/k$. Then
\[
\Inf_i(g)\ge\Inf_i(f)-2\Prob(f\neq g)
\ge \frac1k-\frac{2\operatorname{dist}_H(f,g)}{2^{k-1}}\ge \frac1k - \frac r{2^{k-2}} >0
\]
provided $2^{k-2}/k>r$. For each fixed $r$, this holds for all
sufficiently large $k$, as required by the theorem. Given a leader strategy, every possible
leader, including a newly added subdivision vertex\footnote{Indeed, an ordinary follower can reach a subdivision vertex
within distance $r$ only through the ordinary endpoint of its path (reaching it through the special endpoint would require more than
$r$ steps). If that ordinary endpoint corresponds to $f$, all ordinary
followers therefore correspond to functions at distance at most $r$
from $f$ in $H$, so we may replace this leader by $v_{i_f}$.}, can again be replaced by
a special vertex visible to all its ordinary followers, without
increasing disagreements along ordinary edges. The remainder of the proof is unchanged. Taking $N$ sufficiently
large depending on $k$ and $r$ makes the contribution of the
subdivision paths negligible, so the implicit constants in
both bounds remain universal.

\paragraph{Are leader equilibria always optimal? \\(Short answer: No)}~\\
\\
This part discusses the related game model of coordination on social networks (including the messaging process and payoff convention), for all of which we refer the reader to~\cite{TOOsqrtbound2026}. The realization can also be chosen so that the binary profile is an
action-symmetric Nash equilibrium. Indeed, we show that there exists
a realization for which no vertex can benefit by deviating from the
binary profile discussed in the proof. Let special vertices broadcast
and play their independent fair signs, while all other vertices send
empty messages and use only their prescribed source messages, choosing
an independent fair action if any required message is invalid.

For each fixed $f$ and each $g\sim f$, orient the normal vector
$x_{fg}\in\{\pm1\}^k$ to the hyperplane containing the common facet
between the regions of $f$ and $g$ so that $f(x_{fg})=1$.
The same path argument used in Lemma~\ref{lem:key-inequality} gives
\[
\sum_{g\sim f}p_{fg}x_{fg,i}=\Prob(f_W=f)>0
\qquad(i\in\supp(f)).
\]
Fix an input $x$. If all observed coordinates equal $f(x)$, then each
ordinary neighbor is at least as likely to agree as to disagree,
while all special neighbors agree so $f(x)$ is a best response
in this case. Otherwise, choose $i\in\supp(f)$ such that
$f(x)x_i=-1$. Notice that for every $\{f,g\}\in E(H)$,
$f(x)g(x)\ge x_{fg,i}$. If $f(x)=g(x)$ this is immediate.
If $f(x)\neq g(x)$, then either $x=x_{fg}$ or $x=-x_{fg}$,
so in both cases $x_{fg,i}=f(x)x_i=-1$, and the inequality
again holds.

By concentration, for sufficiently large $N$, the following holds
with high probability, simultaneously for every $u_f^j\in V_f$
and every such input $x$:
\[
\begin{aligned}
\sum_{g\sim f}|\{u_g^\ell\in V_g:u_g^\ell\sim u_f^j\}|\,f(x)g(x)-k
&\ge
\sum_{g\sim f}|\{u_g^\ell\in V_g:u_g^\ell\sim u_f^j\}|\,x_{fg,i}-k\\
&\ge \frac N2\Prob(f_W=f)-k>0.
\end{aligned}
\]
Thus, $f(x)$ is the best response, even after accounting for the 
remaining special neighbors (of which there are at most $k$). Monotonicity and oddness also imply that no special
vertex benefits from changing its action or sending different messages:
broadcasting its chosen action maximizes agreement with each neighbor,
and both actions give the same expected payoff. These conditions hold with high probability
and can therefore be imposed together with the preceding estimates.

\medskip

For $r>1$, the communication model of~\cite[Section~2]{TOOsqrtbound2026}
allows each special vertex to send its sign directly to every vertex
within distance $r$. Each subdivision vertex uses only the message
from the special endpoint of its path and plays that sign, choosing
an independent fair action if the message is invalid. All nonspecial
vertices still send empty messages, which are ignored. Ordinary
vertices have the same relevant observations and neighboring actions
as before. Each subdivision vertex agrees with its neighbor toward
the special endpoint, so changing its action cannot improve its
payoff, while its outgoing messages have no effect. Special vertices
agree with all their neighbors. Thus the equilibrium property is
preserved.

\medskip

Altogether, this provides a negative answer to the question of~\cite{TOOsqrtbound2026} whether $(\varepsilon,r)$-amenability coincides with having $\varepsilon$ inefficiency in equilibria
with a radius of communication $r$.

\begin{remark}
The extension of the preceding equilibrium argument from $r=1$
to $r>1$ uses the communication model of HRT, which allows
messages to be sent directly to any vertex at distance at most
$r$. We remark that the same profile can also be implemented
as a Nash equilibrium using $r$ rounds of communication between
immediate neighbors, followed by a single action stage. Along each subdivided path
$v_i=w_0,w_1,\cdots,w_r=u_f^j$, the special vertex sends its sign
in the first round, and each subdivision vertex forwards the
received message toward $u_f^j$ in the next round and eventually
plays the received sign. Only messages from the preceding vertex
at the prescribed round are used, while all other messages are ignored,
and invalid or missing messages are forwarded as invalid and give
an independent fair action. Ordinary vertices continue to send
empty messages.

The vertices $w_1,\ldots,w_{r-2}$ agree with both neighbors and
therefore already attain their maximum payoff. The remaining
subdivision vertex $w_{r-1}$ observes only $Z_{v_i}$ and its own
independent randomness. For either valid forwarded sign $m_i$,
oddness and monotonicity give
\[
\Prob\!\left(
f(Z_{v_1},\ldots,Z_{v_{i-1}},m_i,Z_{v_{i+1}},\ldots,Z_{v_k})
=m_i \,\middle|\, Z_{v_i}
\right)
=\frac{1+\Inf_i(f)}{2}>\frac12.
\]
Thus, conditional on its information, no joint choice of message
and action can yield a higher agreement probability with $u_f^j$;
an invalid message yields agreement probability $1/2$. Forwarding
and playing $Z_{v_i}$ attains this bound while also agreeing
certainly with $w_{r-2}$. Hence no subdivision vertex benefits
from a joint message and action deviation. Special vertices still
agree with all their neighbors, and the previous best-response
argument for ordinary vertices is unchanged.
\end{remark}

\section{Proof of Proposition~\ref{prop:varaince games}} 
\label{sec:var_game_TV_bound}

Let $Z_1,\ldots ,Z_k$ be independent random variables, and let $f\ (=X_i)$ and $g\ (=X_j)$ be unbiased binary random variables measurable with respect to \(\sigma(Z_1, \dots, Z_k)\). Recall the definition of the Shapley value measures $\mu_f\ (=\mu_i)$ and $\mu_g\ (=\mu_j)$ presented in \eqref{eq:vi-definition} and \eqref{eq:Shapley value measure definition}. Since $2\mathbb P(f\neq g)= 1-\operatorname{Cov}(f,g)\ (=\varepsilon_{ij})$, it is sufficient to show that there exists a universal constant $C>0$ such that
\[
\|\mu_f-\mu_g\|_{\operatorname{TV}}
\le
C\,\mathbb{P}(f\neq g)\left(1+\log\frac{1}{\mathbb{P}(f\neq g)}\right)\, .
\]
We do so by introducing an operator $M_\theta$ for every $\theta \in [-1,1]^k$, for which 
\[
\langle (f-g) , M_\theta (f+g)\rangle = \sum_{i} \theta_i(\mu_f(i)-\mu_g(i)),
\]
whose maximal possible value over all $\theta$'s gives $2\|\mu_f-\mu_g\|_{\operatorname{TV}}$. 
We show that the expression above is bounded, by defining a martingale transform whose expectation over a uniformly chosen permutation gives the value of the operator, and then applying Burkholder's inequality~\cite{Burkholder1984} to get the desired result.  

\medskip

\begin{proof}
Following~\cite[proof of Proposition~1]{TOOsqrtbound2026},
we work in $L^2(Z_K)$. Let $L_i$ be its closed subspace of
mean-zero, $\sigma(Z_i)$-measurable random variables, and let
$\mathcal U_i$ be an orthonormal basis of $L_i$. Since the variables \(Z_i\) are independent, the collection of products
\[
U_S=\prod_{i\in S}U_i,
\qquad U_i\in\mathcal U_i,
\]
where \(\emptyset \neq S\subseteq[k]\), forms an orthonormal basis for the mean-zero subspace of \(L^2(Z_K)\). For each nonempty \(S\subseteq[k]\), denote by \(\mathcal U_S\) the collection of all such products \(U_S\).

Thus we may write
\(
f
=
\sum_{\emptyset\neq S\subseteq[k]}
\sum_{U_S\in\mathcal U_S}
\widehat f(U_S)U_S,
\)
where
\(
\widehat f(U_S):=\langle f,U_S\rangle.
\)

\medskip

Attributing to \cite{Owen2014,OwenPrieur2017}, HRT show that the Shapley-value measure is given explicitly by
\[
\mu_f(i)
=
\sum_{\substack{S\subseteq[k]\\ i\in S}}
\sum_{U_S\in\mathcal U_S}
\frac{\widehat f(U_S)^2}{|S|}\,.
\]

\medskip

Recall that the total variation distance is given by
\[
\|\mu_f-\mu_g\|_{\operatorname{TV}}
=
\frac12
\sum_{i=1}^k \bigl|\mu_f(i)-\mu_g(i)\bigr|
=
\frac12
\sup_{\theta_i\in[-1,1]}
\sum_{i=1}^k \theta_i\bigl(\mu_f(i)-\mu_g(i)\bigr).
\]
Now fix \(\theta=(\theta_1,\dots,\theta_k)\in[-1,1]^k\), and define the operator\footnote{with the convention \(M_\theta 1=0\).}
\[
M_\theta U_S
=
m_\theta(S)U_S,
\text{ where } \
m_\theta(S):=
\frac1{|S|}
\sum_{i\in S}\theta_i
\qquad(S\neq\emptyset).
\]
Notice that \( \sum_i\theta_i\mu_f(i)= \langle f,M_\theta f\rangle,\) which gives us
\[
\sum_i\theta_i\bigl(\mu_f(i)-\mu_g(i)\bigr)
=
\langle f,M_\theta f\rangle-\langle g,M_\theta g\rangle.
\]
Putting \( \varphi:=f-g \),
and noticing the operator \(M_\theta\) is self-adjoint, one may verify
\[
\langle f,M_\theta f\rangle-\langle g,M_\theta g\rangle
=
\langle \varphi,M_\theta f\rangle+\langle g,M_\theta \varphi\rangle
=
\langle \varphi,M_\theta f\rangle+\langle \varphi,M_\theta g\rangle.
\]

Thus we now wish to bound the expressions of the form $|\langle \varphi,M_\theta u\rangle|$ where \(u\) is a $\{1,-1\}$-valued random variable.

\medskip

We prove the following auxiliary lemma.

\begin{lemma}\label{lem:bound on M_theta}
There exists a universal constant \(C>0\) such that for every random variable \(u\in L^2(Z_K)\) with \(|u|\le1\), every \(\theta\in[-1,1]^k\), and every event \(A\),
\[
\int_A |M_\theta u|\,d\mu
\le
C\,\mu(A)\left(1+\log\frac{1}{\mu(A)}\right),
\]
where \(\mu\) denotes the joint law of \(Z_K=(Z_1,\ldots,Z_k)\).
\end{lemma}
\begin{proof}
Fix a permutation $\pi$ of $[k]$, and define
\[
\mathcal F_j^\pi
:=\sigma(Z_{\pi(1)},\ldots,Z_{\pi(j)}),
\qquad
u_j^\pi:=\mathbb E[u\mid\mathcal F_j^\pi],
\]
where $\mathcal F_0^\pi$ is trivial and $u_0^\pi=\mathbb E u$.
The corresponding martingale transform~\cite{martingaleTransforms} is
\[
T_{\theta,\pi}u
:=\sum_{j=1}^k\theta_{\pi(j)}
   (u_j^\pi-u_{j-1}^\pi).
\]
We notice that
\[
\mathbb E[U_S\mid \mathcal F_j^\pi]
=
\begin{cases}
0, & S\not\subseteq\{\pi(1),\dots,\pi(j)\},\\
U_S, & S\subseteq\{\pi(1),\dots,\pi(j)\}.
\end{cases}
\]
Indeed, if some coordinate in \(S\) has not yet been revealed, then the corresponding factor has mean zero and is independent of \(\mathcal F_j^\pi\). If all coordinates in \(S\) have been revealed, then \(U_S\) is \(\mathcal F_j^\pi\)-measurable. Thus, if \(i\in S\) is the last coordinate in \(S\) that is revealed by \(\pi\), then the element $u_{\pi^{-1}(i)}^\pi - u_{\pi^{-1}(i)-1}^\pi$ (with respect to \(U_S\)) is the only nonzero martingale difference, so
\[
T_{\theta,\pi}U_S
=
\theta_i U_S.
\]
Under a uniformly chosen permutation, each $i\in S$ is equally
likely to be last, so
\[
\mathbb E_\pi[T_{\theta,\pi}U_S]
=\frac1{|S|}\sum_{i\in S}\theta_iU_S
=M_\theta U_S.
\]
Both operators annihilate constants and are bounded on $L^2$.
By linearity and $L^2$ approximation, therefore,
\[
M_\theta u=\mathbb E_\pi[T_{\theta,\pi}u].
\]

Burkholder's martingale transform inequality~\cite{Burkholder1984}
gives, for every $q\ge2$,
\[
\|T_{\theta,\pi}u\|_q
\le(q-1)\|u-\mathbb E u\|_q
\le2(q-1),
\]
since $|\theta_i|\le1$ and $|u|\le1$. Averaging over permutations
and applying the triangle inequality yields
\[
\|M_\theta u\|_q
\le\mathbb E_\pi\|T_{\theta,\pi}u\|_q
\le2(q-1).
\]
If $\mu(A)=0$, the claim is immediate. Otherwise, choose
\[
q=\max\left\{2,\log\frac1{\mu(A)}\right\},
\]
where the logarithm is natural. H\"older's inequality gives
\[
\begin{aligned}
\int_A|M_\theta u|\,d\mu
&\le \mu(A)^{1-1/q}\|M_\theta u\|_q\\
&\le 2(q-1)\mu(A)^{1-1/q}\\
&\le 2e\,\mu(A)\left(1+\log\frac1{\mu(A)}\right),
\end{aligned}
\]
where the last inequality uses
$\mu(A)^{-1/q}\le e$ and
$q-1\le1+\log(1/\mu(A))$.
\end{proof}

\medskip

We return to the proof of the Proposition. Let
$A:=\{x:f(x)\neq g(x)\}$. Then \(\mu(A)=\mathbb{P} (f \neq g)\), and \(|\varphi|=|f-g|=2\cdot \mathbf 1_A.\) 
Therefore, for every \(u\) with \(|u|\le1\), we have
\[
|\langle \varphi,M_\theta u\rangle|
=
\left|
\int \varphi\,M_\theta u\,d\mu
\right|
\le
2\int_A |M_\theta u|\,d\mu \le C{\mu(A)}\left(1+\log\frac1{\mu(A)}\right)
\]
by the lemma. Applying this with \(u=f,g\), gives (together with the first part of the proof) that
\[
\left|
\sum_i\theta_i(\mu_f(i)-\mu_g(i))
\right|
\le
C{\mu(A)}\left(1+\log\frac1{\mu(A)}\right).
\]
Thus
\[
\|\mu_f-\mu_g\|_{\operatorname{TV}}
\le
C{\mu(A)}\left(1+\log\frac1{\mu(A)}\right)=C\,\mathbb{P}(f\neq g)\left(1+\log\frac{1}{\mathbb{P}(f\neq g)}\right).
\]
This estimate is equivalent to the stated one (up to constants), thus completing the proof of the Proposition.
\end{proof}

\section{Additional proofs}\label{sec:additional proofs}
\begin{proof}[Proof of Proposition~\ref{sqrt is sharp in general}]
Consider the cycle graph \(C_n\), with \(n\) much larger than \(r\). Let \(Z_1,\ldots,Z_n\) be independent random variables taking values \(\pm1\) with probability \(1/2\).

For each \(i\in C_n\), define
\[
Y_i:=\sum_{j=-r}^r (r+1-|j|)Z_{i+j},
\]
where indices are taken modulo \(n\), and put \(X_i:=Y_i/\|Y_i\|_2\). Then \(X_i\) is measurable with respect to \(B_r(i)\), and
\(\mathbb E X_i=0\), \(\mathbb E X_i^2=1\).

\medskip

We have
\[
\|Y_i\|_2^2=\sum_{j=-r}^r (r+1-|j|)^2=(r+1)^2+2\sum_{k=1}^r k^2.
\]
In particular, \(\|Y_i\|_2^2\ge r(r+1)(2r+1)/3\).
Also notice that the coefficient vector of \(Y_i-Y_{i+1}\) simply has \(2r+2\) nonzero coefficients, each equal to \(+1\) or \(-1\) (see Figure~\ref{fig:placeholder}), so by independence
\[ \mathbb E[(Y_i-Y_{i+1})^2]=2r+2.\]
Therefore,
\[
\mathbb{E}[(X_i-X_{i+1})^2] =\frac{2r+2}{\|Y_i\|_2^2} \le \frac{6}{r(2r+1)} \le \frac{6}{r^2}.
\]

\begin{figure}
    \centering
    \includegraphics[width=0.8\linewidth]{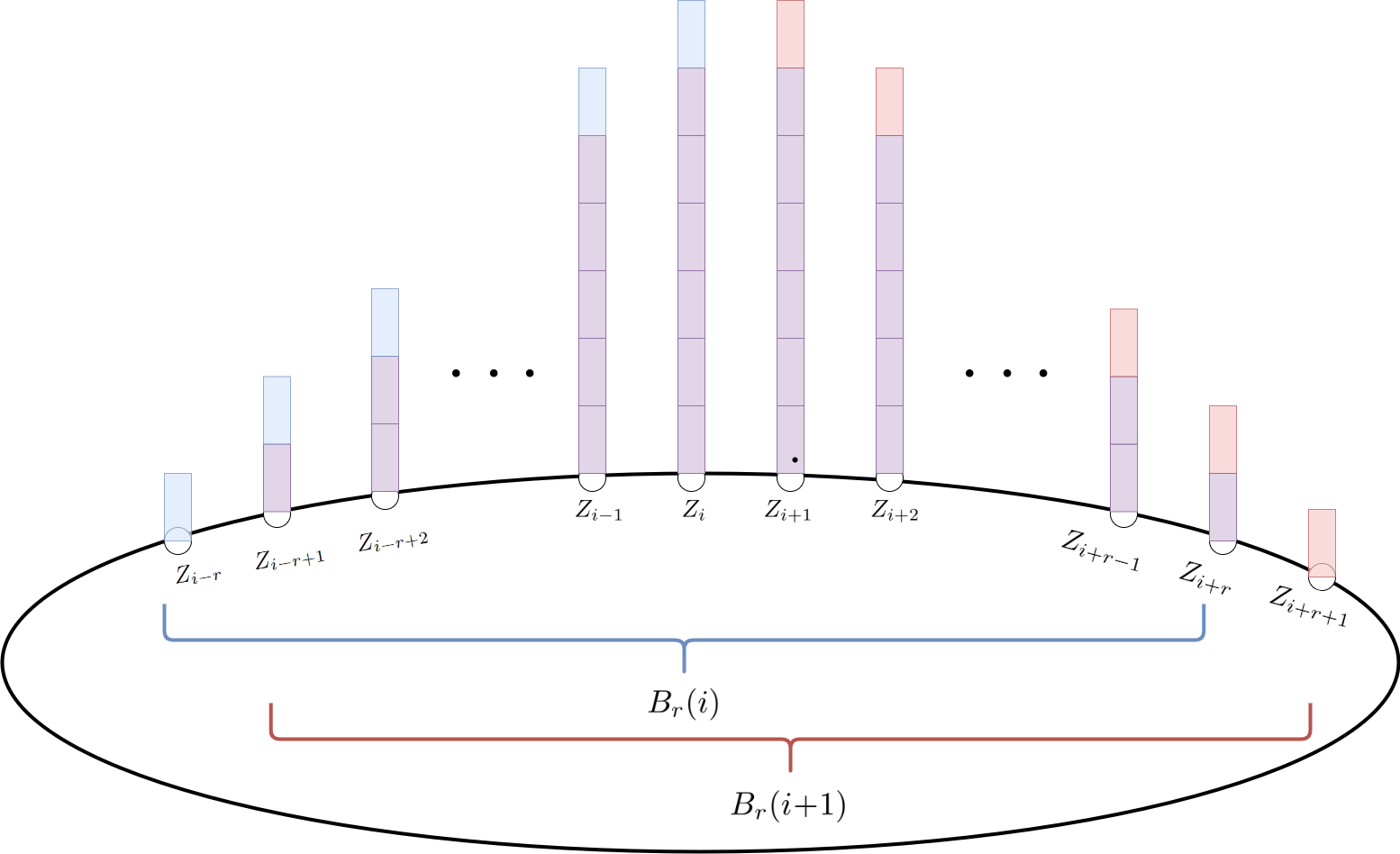}
    \caption{The weights corresponding to the variable $Y_i$ are marked \color{cyan} blue\color{black}. The weights corresponding to the variable $Y_{i+1}$ are marked \color{red}red\color{black}. The overlap is marked \color{violet}purple\color{black}.}
    \label{fig:placeholder}
\end{figure}

Thus the inefficiency satisfies
\[
\varepsilon_r:=\frac1{|E(C_n)|}\sum_{\{i,j\}\in E(C_n)}\frac12\mathbb E[(X_i-X_j)^2]
\le \frac{3}{r^2}.
\]

On the other hand, every subset of \(C_n\) of radius at most \(r\) has size at most \(2r+1\). Hence any partition of \(C_n\) into radius-\(r\) parts has at least \(n/(2r+1)\) parts, and therefore cuts at least order \(n/r\) edges. Since \(|E(C_n)|=n\), the best possible amenability parameter is at least \(c/r\) for some universal \(c>0\), provided \(n\) is sufficiently large compared with \(r\). Taking $r=\lceil\sqrt{3/\varepsilon}\rceil+1$ gives the claimed sharpness (up to constant) of the square-root dependence.
\end{proof}

\medskip

\begin{proof}[Proof of Proposition~\ref{prop:coreH dist}]
Let \(K=\{1,\ldots,k\}\), let \(Z_1,\ldots,Z_k\) be independent random variables, and let \(X\) be an unbiased sign random variable measurable with respect to \(Z_K\). Following the decomposition used by HRT~\cite{TOOsqrtbound2026}, we may write
\[
X=\sum_{\emptyset\neq A\subseteq K}X_A,
\]
where \(X_A\) is measurable with respect to \(Z_A\), and the random variables \(X_A\) are pairwise orthogonal.\footnote{Explicitly, setting \(X_{\emptyset}=0\), we recursively define \(X_A=\mathbb E[X\mid Z_A]-\sum_{B\subsetneq A}X_B\).}

\medskip

Recall the variance TU game associated with \(X\),
\[
v_X(S):=\operatorname{Var}(\mathbb E[X\mid Z_S])
       =\sum_{\emptyset\neq A\subseteq S}\|X_A\|_2^2,
\qquad S\subseteq K.
\]
Since \(X\) is an unbiased sign random variable, \(v_X(K)=\operatorname{Var}(X)=1\).

\medskip

We consider the core of this game,
\[
\operatorname{Core}(v_X)
:=
\left\{
\mu\in\Delta(K):
\mu(S)\geq v_X(S)\text{ for every }S\subseteq K
\right\}.
\]
Let \(\lambda\in\Delta(K)\). Consider the directed network
\[
\Gamma:
\qquad
V(\Gamma)=\{v_{\mathrm{source}},v_{\mathrm{sink}}\}
\cup\{A:\emptyset\neq A\subseteq K\}\cup K,
\]
with edge set
\[
E(\Gamma)
=
\{(v_{\mathrm{source}},A):\emptyset\neq A\subseteq K\}
\cup\{(A,a):a\in A\}
\cup\{(a,v_{\mathrm{sink}}):a\in K\}.
\]
The edge \((v_{\mathrm{source}},A)\) has capacity \(\|X_A\|_2^2\), the edge \((a,v_{\mathrm{sink}})\) has capacity \(\lambda(a)\), and every edge \((A,a)\), \(a\in A\), has infinite capacity. See Figure~\ref{fig:flowchart}.

\begin{figure}
    \centering
    \includegraphics[width=0.6\linewidth]{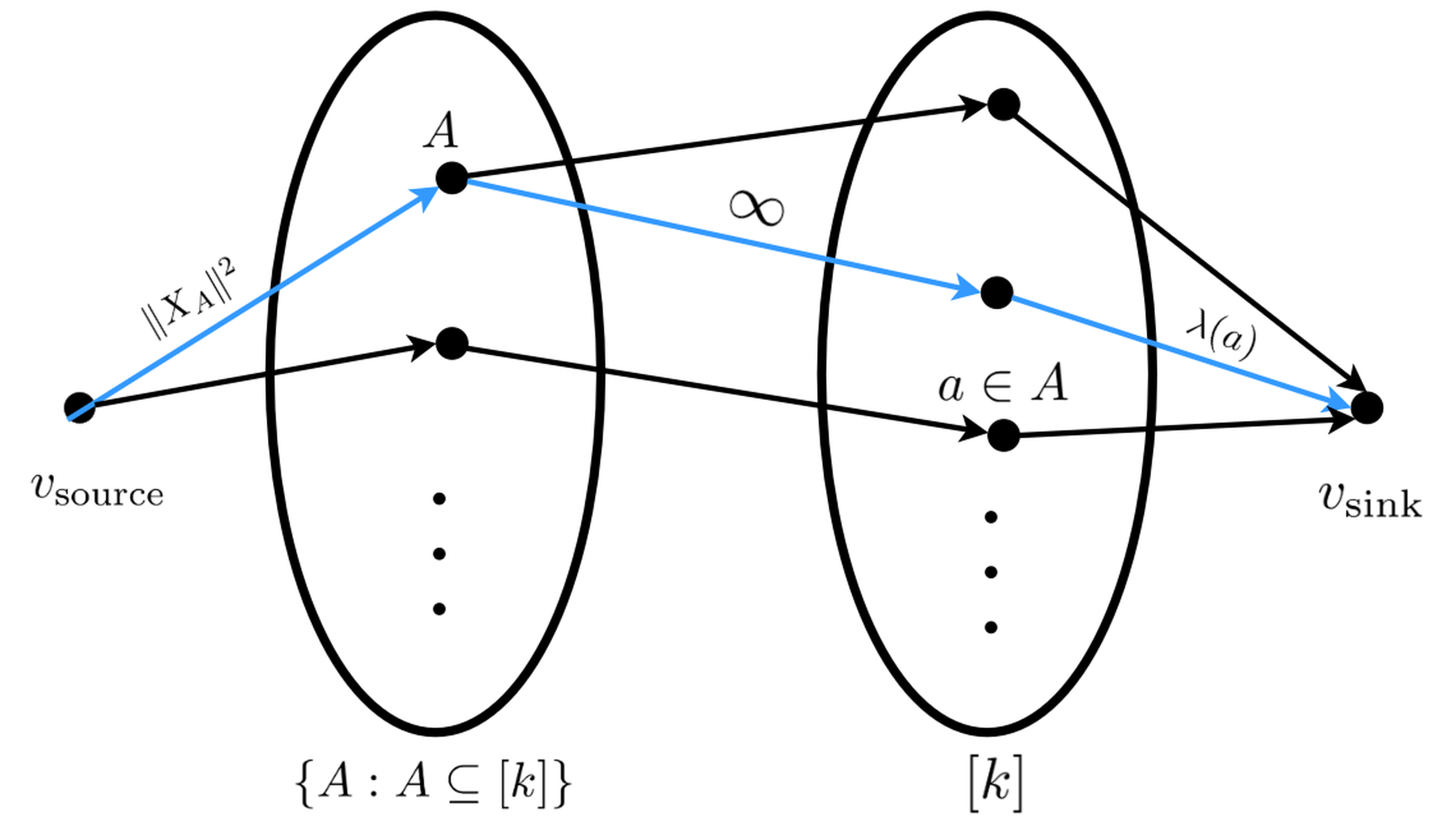}
    \caption{The flow in the network \(\Gamma\).}
    \label{fig:flowchart}
\end{figure}

Let \(F\) be a flow in \(\Gamma\), and denote by \(F_{A,a}\) the flow through the edge \((A,a)\). The capacity restrictions imply
\[
\sum_{a\in A}F_{A,a}\leq \|X_A\|_2^2
\qquad\text{for every }\emptyset\neq A\subseteq K,
\]
and
\[
\sum_{A\ni a}F_{A,a}\leq\lambda(a)
\qquad\text{for every }a\in K.
\]

The flow \(F\) distributes part of the mass \(\|X_A\|_2^2\) associated with each set \(A\) among the elements of \(A\). We now construct a probability measure \(\mu\) by distributing the unused mass of each \(A\) arbitrarily among the elements of \(A\). In other words, for each nonempty \(A\subseteq K\), choose \(q_A\in\Delta(A)\) so that the flow already assigned to every \(a\in A\) is respected, and write
\[
\mu(a)=\sum_{\emptyset\ne A\subseteq K}q_A(a)\|X_A\|_2^2.
\]
Then, for every \(S\subseteq K\),
\[
\mu(S)
=
\sum_{\emptyset\ne A\subseteq K}q_A(S)\|X_A\|_2^2
\geq
\sum_{\emptyset \ne A\subseteq S}\underbrace{q_A(S)}_{=1}\|X_A\|_2^2
=
v_X(S).
\]
Thus \(\mu\in\operatorname{Core}(v_X)\).\footnote{In fact, this holds for every choice of \(q_A\in\Delta(A)\). In particular, taking \(q_A(a)=1/|A|\) for \(a\in A\) yields the Shapley value measure, showing directly that the Shapley value measure belongs to \(\operatorname{Core}(v_X)\).}

Moreover, for every \(a\in K\),
\[
\sum_{A\ni a}F_{A,a}\leq\lambda(a)
\qquad\text{and}\qquad
\sum_{A\ni a}F_{A,a}\leq\mu(a),
\]
since the construction of \(\mu\) only adds mass to that already assigned by the flow. Hence
\[
|F|
=
\sum_{a\in K}\sum_{A\ni a}F_{A,a}
\leq
\sum_{a\in K}\min\{\lambda(a),\mu(a)\}
=
1-\|\mu-\lambda\|_{\operatorname{TV}}.
\]
Here we used the elementary identity
\[
\|\mu-\lambda\|_{\operatorname{TV}}
=
\frac12\sum_{a\in K}|\mu(a)-\lambda(a)|
=
1-\sum_{a\in K}\min\{\mu(a),\lambda(a)\}.
\]
It follows that
\[
\min_{\mu\in\operatorname{Core}(v_X)}
\|\lambda-\mu\|_{\operatorname{TV}}
\leq
1-|F_{\max}|,
\]
where \(F_{\max}\) is a maximum flow in \(\Gamma\).

\medskip

To calculate the max flow we search for the minimal sum of capacities of a cut, and then invoke the min-cut-max-flow Theorem.
Let $S\subseteq[k]$ be the set of element vertices on the source
side of a minimum cut. The edges from $S$ to $v_{\mathrm{sink}}$
contribute $\lambda(S)$ to its capacity. If a subset vertex $A$
is on the source side, then $A\subseteq S$, since otherwise an
edge of infinite capacity would cross the cut. Conversely, if
$A\subseteq S$, moving $A$ to the source side cannot increase
the cut capacity. We may therefore choose a minimum cut in
which $A$ is on the source side exactly when $A\subseteq S$.

Consequently, by the max-flow min-cut theorem,
\[
|F_{\max}|
=
\min_{S\subseteq K}
\left(
\lambda(S)+
\sum_{A\nsubseteq S}\|X_A\|_2^2
\right).
\]
Since \(\sum_{\emptyset\neq A\subseteq K}\|X_A\|_2^2=1\), this becomes
\[
|F_{\max}|
=
\min_{S\subseteq K}
\left(
\lambda(S)+1-
\sum_{A\subseteq S}\|X_A\|_2^2
\right)
=
1-\max_{S\subseteq K}\bigl(v_X(S)-\lambda(S)\bigr).
\]
We have therefore proved
\[
\min_{\mu\in\operatorname{Core}(v_X)}
\|\lambda-\mu\|_{\operatorname{TV}}
\leq
\max_{S\subseteq K}\bigl(v_X(S)-\lambda(S)\bigr).
\]

\medskip

Now let \(X,Y\) be unbiased sign random variables measurable with respect to \(Z_K\), and let \(\lambda\in\operatorname{Core}(v_Y)\). Then
\[
\min_{\mu\in\operatorname{Core}(v_X)}
\|\lambda-\mu\|_{\operatorname{TV}}
\leq
\max_{S\subseteq K}\bigl(v_X(S)-\lambda(S)\bigr)
\leq
\max_{S\subseteq K}\bigl(v_X(S)-v_Y(S)\bigr),
\]
where the second inequality follows from \(\lambda(S)\geq v_Y(S)\) for every \(S\subseteq K\).

As in HRT, for every \(S\subseteq K\),
\[
\begin{split}
|v_X(S)-v_Y(S)|
&=
\left|
\mathbb E\left[
\mathbb E[X\mid Z_S]^2-\mathbb E[Y\mid Z_S]^2
\right]
\right|\\
&=
\left|
\mathbb E\left[
\bigl(\mathbb E[X\mid Z_S]-\mathbb E[Y\mid Z_S]\bigr)
\bigl(\mathbb E[X\mid Z_S]+\mathbb E[Y\mid Z_S]\bigr)
\right]
\right|\\
&\leq
2\,\mathbb E\left|
\mathbb E[X-Y\mid Z_S]
\right|
\leq
2\,\mathbb E|X-Y|
=
4\mathbb P(X\neq Y).
\end{split}
\]
Therefore,
\[
\sup_{\lambda\in\operatorname{Core}(v_Y)}
\inf_{\mu\in\operatorname{Core}(v_X)}
\|\lambda-\mu\|_{\operatorname{TV}}
\leq
4\mathbb P(X\neq Y).
\]
Interchanging the roles of \(X\) and \(Y\) gives
\[
d^{\operatorname{H}}_{\operatorname{TV}}
\bigl(
\operatorname{Core}(v_X),
\operatorname{Core}(v_Y)
\bigr)
\leq
4\mathbb P(X\neq Y).
\]    
\end{proof}

\medskip

\begin{proof}[Proof of Corollary~\ref{cor:tree-linear-bound}]
    Given a binary coordination profile $\pi$, root the tree, choose any
measure in the root's core, and recursively choose a measure in each
child's core closest (in total variation distance) to its parent's choice. Proposition~\ref{prop:coreH dist}
ensures that
$\|\mu_i-\mu_j\|_{\operatorname{TV}}\le4\mathbb P(X_i\neq X_j)$
on every edge.
The grand coupling theorem then gives a leader profile of cost at most
\[
\frac{8}{|E|}\sum_{\{i,j\}\in E}\mathbb P(X_i\neq X_j)
=4\,\mathtt C(T,\pi),
\]
and taking the infimum over $\pi$ proves the claim.
\end{proof}

\medskip

\begin{proof}[Proof of Proposition~\ref{prop:C(r,Delta)}]
We prove the proposition with $C(r,\Delta)=8\Delta^{8r+4}$. 
Let $G=(N,E)$ be a graph of maximal degree at most $\Delta>1$ and $r>0$ an integer (the case $\Delta =1$ is trivial). Let $\pi=(X_i,Z_j)_{i,j\in N}$ be a coordination profile adapted to $\mathfrak S(G,r)$.

Denote
\[
c_{ij}:=1-\operatorname{Cov}(X_i,X_j)
\qquad\text{and}\qquad
M:= \max_{i\in N}|B_{2r}(i)|<\Delta^{2r+1}\, .
\]
Set a threshold \(\tau=\frac{1}{C(r,\Delta)}\) 
and delete every edge ${i,j}$ with $1-\operatorname{Cov}(X_i,X_j)=c_{ij}\geq\tau$. Since
\[
\sum_{\{i,j\}\in E(G)}1-\operatorname{Cov}(X_i,X_j)=\mathtt C(G,\pi)|E|,
\]
at most $C(r,\Delta)\,\mathtt C(G,\pi)|E|$ edges are deleted.

Let $D$ be the set of deleted edges and write $G'=(N,E\setminus D)$. Let $C$ be a connected component of $G'$. We work with $G'[C]$, keeping the information sets $B_r(i)$ defined in the original graph $G$. If $C$ is a singleton, its vertex can follow itself, so we assume that $C$ contains at least two vertices. Fix $i_0\in C$. We claim that $C\subseteq B_{2r}(i_0)$. Otherwise, take a shortest path $\gamma=(i_0,\dots,i_\ell)$ in $G'[C]$ from $i_0$ to $C\setminus B_{2r}(i_0)$. We have $\operatorname{dist}_{G}(i_0,i_\ell)=2r+1$ and $\ell\leq M$. Also, $B_r(i_0)\cap B_r(i_\ell)=\emptyset$, so $X_{i_0}$ and $X_{i_\ell}$ are independent and therefore $\|X_{i_0}-X_{i_\ell}\|_2=\sqrt 2$. On the other hand, 
\[
\sqrt 2 = \|X_{i_0}-X_{i_\ell}\|_2 \le \sum_{j=1}^\ell \| X_{i_j} - X_{i_{j-1}} \|_2 < \ell \sqrt {2 \tau}  < \sqrt 2\, ,
\]

a contradiction. Hence $|C|\leq M$.

Moreover, $\pi_{|C}:=(X_i,Z_j)_{i\in C,j\in N}$ is a coordination profile adapted to $\mathfrak S_C:=(C,N,(B_r)_{|C})$ with discoordination cost  $\mathtt C(G'[C],\pi_{|C})< \tau$, since every edge of \(G'[C]\) has cost \(1-\operatorname{Cov}(X_i,X_j)<\tau\).
 By \cite[Theorem 3]{TOOsqrtbound2026}, 
there exists a leader coordination profile ${\pi_{\mathrm{leader}}}_{|C}$ adapted to $\mathfrak S_C$ with $\mathtt C(G'[C],{\pi_{\mathrm{leader}}}_{|C})\leq \sqrt{8\tau}<1/M^{2}$. Since $G'[C]$ has at most $\binom{M}{2}<M^2$ edges, this leader profile has
zero disagreement on every edge of $G'[C]$. By connectedness, all
vertices of $C$ therefore have the same output almost surely, and
hence a common available information source. Choose such a source
as the leader for each component. Combining these choices gives a
leader coordination profile adapted to $\mathfrak S(G,r)$ whose
only possible disagreements are on the deleted edges. All such
edges have already been counted in $D$, including those with both
endpoints in the same component. Therefore, 
\[
\mathtt I_{\mathrm{leader}}(G,r)\leq \mathtt C(G,\pi_{\mathrm{leader}})\leq C(r,\Delta)\,\mathtt C(G,\pi)\, .
\]
Since this holds for every coordination profile adapted to $\mathfrak S(G,r)$, we conclude
\[
\mathtt I_{\mathrm{leader}}(G,r) \le C(r,\Delta)\mathtt I (G,r)\, .
\]
\end{proof}

\section{Acknowledgements.}
\noindent\textbf{AI methodology.} ChatGPT 6 was used to develop a simpler proof of Lemma~\ref{lem:bound on M_theta}, which is used in the proof of Proposition~\ref{prop:varaince games}. The original proof of the lemma relied on \say{heavy machinery}, such as the BMO norm, the John–Nirenberg inequality, and the Hardy–Littlewood inequality~\cite{originalBMO,discreteJohnNirenberg,HardyLittlewoodPolya1952}, whereas the revised proof does not require these tools and is more concise and elegant. ChatGPT 6 was also used to compute the sum of probabilities in Lemma~\ref{lem:sum-probabilities}, which is used in the proof of Theorem~\ref{thm:counter to linearity}. The authors independently verified the correctness and originality of all content, including the references.

\medskip
\noindent\textbf{Grant support.} Ron Peretz and Dean Kraizberg gratefully acknowledge the support of the Israel Science Foundation Grant~$\#1461/26$.

\end{document}